\documentclass{article}[12pt]

\usepackage{amsmath,amssymb,amsthm,amsfonts}
\usepackage{times}
\usepackage{enumerate}
\usepackage{epsfig}
\usepackage{graphicx}
\usepackage{pgf}
\usepackage{chemarr}
\usepackage{float}
\restylefloat{table}

\usepackage[linesnumbered,ruled,commentsnumbered,noline]{algorithm2e}

\usepackage{multirow}
\usepackage[left]{lineno}
\usepackage{hhline}
\usepackage{lipsum}

\newtheorem{theorem}{Theorem}[section]
\newtheorem{lemma}[theorem]{Lemma}

\newtheorem{remark}[theorem]{Remark}

\def \non{{\nonumber}}
\def \tilde{\widetilde}
\def \hat{\widehat}

\newcommand{\f}{\frac}

\newcommand{\rt}{\rightarrow}

\newcommand{\np}{\noindent}

\newcommand{\hs}{\hspace}
\newcommand{\vs}{\vspace}

\newcommand{\s}{\sigma}
\def\SC{\mathcal}

\def \triple|{|\! | \! |}

\def\le{\left}
\def\ri{\right}

\def\l{\lambda}

\def\<{\langle}
\def\>{\rangle}
\def\~{\tilde}
\newcommand{\EE}{\ensuremath{\mathsf{E}}}
\newcommand{\PP}{\ensuremath{\mathsf{P}}}

\def\N{\mathbb N}
\def\Z{\mathbb Z}
\def\R{\mathbb R}

\def\Nreact{R} 
\def\Nsp{M} 

\title{Jump-Diffusion Approximation of 
Stochastic Reaction Dynamics: Error bounds and Algorithms} 
\author{Arnab Ganguly$^{1}$\thanks{authors with equal contribution}, \ Derya Alt{\i}ntan $^{2,3,4}$\footnotemark[1], \ Heinz Koeppl$^{4}$}

\begin{document}
\maketitle
\footnotetext[1]{Department of Mathematics, University of Louisville, arnab.ganguly@louisville.edu}
\footnotetext[2]{ Department of Mathematics, Sel\c{c}uk University, altintan@selcuk.edu.tr}
\footnotetext[3]{Department of Scientific Computing, Middle East Technical University}
\footnotetext[4]{Department of Electrical Engineering and Information Technology, Technische Universit\"{a}t Darmstadt,\\ \hs*{0.5cm} heinz.koeppl@bcs.tu-darmstadt.de}
\begin{abstract}
Biochemical reactions can happen on different time scales and also 
the abundance of species in these reactions can be very different from each other. 
Classical approaches, such as deterministic or  stochastic approach 
fail to account for or to exploit  
this multi-scale nature, respectively. In this paper, we propose 
a jump-diffusion approximation for multi-scale  Markov jump  processes that couples the 
two modeling approaches. An error bound of the proposed approximation 
is derived and used to partition the reactions into fast and slow sets, where  the 
fast set is simulated by a stochastic 
differential equation and the slow set is modeled by a discrete chain. 
The error bound leads to a very efficient dynamic partitioning algorithm which
has been implemented for  several multi-scale reaction systems. The gain in computational efficiency 
is illustrated by a realistically sized model of a signal transduction cascade coupled to a gene expression
dynamics. \\


\end{abstract}

\vspace{.1in}

\section{Introduction}

\label{sec:intro}
\vs{.2cm}

A biochemical reaction system involves multiple chemical reactions and several molecular species. 
Recent advances in single cell and single molecule imaging together with microfluidic techniques  have testified to the random 
nature of gene expression and protein abundance in single cells \cite{ccsl:13,ee:10,fcs:10,yxrlx:06}. 
The stochastic nature of  a well-mixed biochemical reaction system is most often captured by modeling the dynamics of  its species' abundance as a continuous time Markov chain (CTMC) 
\cite{ ak:11}. 
Popular algorithms for  exact simulations of such  reaction systems include Gillespie's first and next reaction method \cite{gb:00} and its more efficient variants  \cite{gill:92}. 
These algorithms track every molecular reaction event and thus 
become computationally expensive when reaction system becomes  larger or more complex. 
For reaction systems involving fast reactions and high copy number of species, 
substantial gain in simulation speed is often obtained by resorting to approximate algorithms like tau-leaping  or Langevin (diffusion) approximation. 
However, chemical reactions in biological cells occur with varying orders of abundance of molecular species and varying orders of magnitudes of the reaction rates. 
In these scenarios, suitable hybrid methods need to be implemented to gain speed and efficiency while maintaining the certain level of accuracy. 

Typically, in a hybrid model significant computational efficiency is obtained by treating  some appropriate species as continuous variables and the 
others as discrete ones. The first step in this approach involves partitioning the reaction set into ``fast'' and  ``slow'' reactions. 
The reason for this partitioning is to simulate the fast set either by Langevin (diffusion) approximation or by ordinary differential equation (ODE) approximation while keeping the discrete Markov chain formulation for the slow ones. 
The resulting approximate algorithms give rise to hybrid stochastic  processes where the species with high copy numbers are treated as continuous variables, while the ones with low copy numbers are kept as discrete variables.  
The dynamics of the continuous variables then can be seen to be governed by  ordinary differential equations or stochastic differential equations punctuated by jumps due to the discrete components.

Based on that idea, different hybrid models have been proposed  \cite{cdr:09,hr:02,sk:05}. For example, in \cite{sk:05},  
authors separate the reactions  into fast and slow groups 
such that the Langevin equation is  used 
to simulate the dynamics of fast reactions while integral 
form of the next reaction method is used 
to describe the behavior of the slow ones. 
\cite{hwkt:13} applies  a method of conditional moments (MCM) 
which uses a moment based description for 
the species with high copy number of molecules 
while a stochastic description  is kept for species 
with low copy number of molecules.  
A hybrid representation, that assumes a continuous
and deterministic  behavior 
of the conditional expectation of high copy  species 
given 
the state of  species with low copy numbers, is  introduced in \cite{hmmw:10}. In  \cite{actvh:005}, authors proposed three different algorithms 
 for simulating the hybrid systems  that solve deterministic equations and trace the stochastic reaction event in the meantime.
 Different hybrid strategies for solving 
the chemical master equation are proposed in \cite{hl:07,jah:11,mlsh:12}.



Jahnke and Kreim considered in \cite{jk:12}  a piecewise  deterministic model where species with low copy numbers  are considered as discrete stochastic variables and the species with high copy numbers are treated as considered as continuous variables. In such a model a CTMC process describing the evolution of species with low copy number is coupled with  ODEs representing the dynamics of high copy species. \cite{jk:12} studied the partial thermodynamic limit of such a system and proved that after suitable scaling the approximate error of the hybrid model  in the marginal distributions is of the order $M^{-1}$, where $M$ is a scaling parameter of the system. The parameter $M$ captures the abundance of  high copy  species and is typically chosen such that those species are $O(1)$. However, the partition of the species set is still subjective and it is not immediately clear how to use the result in \cite{jk:12} to form an objective measure for partitioning the species set. Also, it should be noted that the partitioning the species set will have the following effect: a reaction which affects a discrete species and a continuous species will be treated differently in the equations describing the dynamics of the two species. While for the discrete species the number of occurrence of such a reaction will be considered as a stochastic counting process, it will be calculated by a deterministic integral for the continuous species. This might slow  down an algorithm which is based on such a partitioning of the species set, as the same quantity is calculated in two different ways.  In contrast, the present article considers a hybrid diffusion model where the reaction set is partitioned into slow and fast reactions and most importantly, our approach to the error analysis has the sole goal of devising an objective measure for partitioning the reaction set. Our result is formulated in an efficient hybrid algorithm which itself is able to do the partitioning of the reaction set and also check the validity of the partitioning dynamically over the course of time. 


Intuitively, the reaction set with higher propensities 
will occur at a greater speed and using a diffusion
approximation to simulate the occurrence of those reactions 
will preserve the accuracy of the simulation. 
However, in existing literature the identification of reactions with higher 
propensities is often done in a subjective and ad hoc way; 
one difficulty in the designation process lies within varying 
magnitudes of different propensities. 
The higher value of a propensity of a reaction may occur due to its  large 
rate constant  or due to high copy number 
of the reactant species or due to some combination of both the factors.   
The present article attempts to solve this problem by 
introducing a scaling parameter $N$ and suitable scaling 
exponents $\alpha_k$, $\beta_k$ to capture the order of variation of  
the species abundance  and magnitudes of  the reaction rates. 
These types of scaled models were previously 
studied in \cite{KK13} where the authors used 
limiting arguments and 
stochastic averaging techniques for model reduction (also see \cite{KK14}).
The most significant portion of the present paper is a rigorous 
error analysis aimed toward proper identification of the partitioning 
of the reaction set into the fast and the slow ones for a given tolerance for error.

The main theoretical result obtaining the required error bound has been described in Theorem \ref{thm:diffussion}.
The appearance of the different scaling 
exponents $\alpha_k$ and $\beta_k$ in the error bound singles out the reactions 
whose occurrences when simulated by diffusion 
approximation give the lowest possible error.  The methodology forms the backbone of a very accurate dynamic partitioning algorithm  described in Algorithm \ref{sumalgo}. 
While most of the previous works on hybrid simulation 
were based on the chemical master equation \cite{hwkt:13,mlsh:12}, 
the present paper uses an approach based on a  representation 
of the state vector by stochastic equations. These types of differential equation representations of the state of the system \cite{ak:11, EK86} give deeper insight into the full trajectories of both the exact and the approximating processes in contrast to a chemical master type equation, which only describes the state probabilities at  specific time points. The pathwise representations of the processes also allow us to define the error of approximation in a suitable rigorous way, and the corresponding error bound is then derived by proper  use of 
techniques from stochastic analysis.

The rest of the paper is organized as follows. In Section \ref{sec:model}, we describe the Markov chain formulation of the reaction system and the approximating hybrid diffusion model. The pathwise representation of both, the exact and the approximating processes through appropriate stochastic equations are also given. The section also introduces the important scaling exponents required for describing a multi-scale model, and the main error bound is  obtained in Theorem \ref{thm:diffussion}.  
Section \ref{sec:algorithm} concerns itself with the development of the dynamic partitioning algorithm by utilizing the novel
 error bound. Section 
\ref{sec:runge_kutta} describes a  Runge-Kutta  integration method 
for simulating the hybrid diffusion equation, and   
 Section \ref{sec:conservation} makes proper use of the naturally occurring conservation relations in  reaction systems to reduce the dimensionality of the system state and to make the algorithm numerically more robust. Finally, in Section \ref{sec:application}, the proposed algorithm is implemented to analyze 
the Michaelis-Menten kinetics,  the Lotka-Volterra model and 
 the important MAPK pathway together with its gene expression.

\section{Model Setup and Error Bound}
\label{sec:model}
\vs{.2cm}
We will consider chemical reaction systems consisting of $\Nsp$ chemical species, $S_1,\hdots,S_M$, $\Nreact$ reactions, \(R_1,\hdots,R_R\)
\begin{equation}
\label{m_1}
R_k: \quad \sum_{i=1}^\Nsp\nu_{ik}^- S_i \rt \sum_{i=1}^\Nsp\nu^+_{ik}S_i, \quad k=1,\hdots,\Nreact.
\end{equation}
Here, $\nu^-_{ik}$ and $\nu^+_{ik}$ respectively 
denote the number of molecules of the species $S_i$ 
consumed and created  due to one occurrence of reaction $R_k$.
Let $X(t) \in \N^\Nsp$ denote the state of the 
reaction system at time $t$. If $\nu_k$ denote the vector 
with $i$-th component $\nu^+_{ik} -\nu^-_{ik},$ 
then an occurrence of $R_k$ at time $t$ 
updates the state by the following equation
$$X(t) = X(t-) + \nu_k.$$
The process $X$ is a CTMC with transition probabilities governed by
$$\PP[X(t+h) = x+\nu_k|X(t) = x] = a_k(x)h +o(h),$$
where $a_k$ is the propensity function associated 
with reaction $R_k$ and is calculated by 
the law of mass action kinetics in the present article. 
In other words,
\begin{align*}
a_k(x)=  c_k \prod_{i=1}^\Nsp {x_i \choose \nu^-_{ik}},
\end{align*}
where  $c_k$ is the stochastic reaction rate constant for reaction $R_k$. A pathwise representation of the process $X$ 
 is given by the following stochastic equation
\begin{equation}
\label{RTM0}
X(t) = X(0) + \sum_{k=1}^\Nreact\xi_k(\int_0^t a_k(X(s)) ds)\nu_k,
\end{equation}
where the $\xi_k$ are independent unit Poisson processes. It should be noted that the quantities $\xi_k(\int_0^t a_k(X(s)) ds)$ count the number of occurrences of the reaction $R_k$ by time $t$. \eqref{RTM0} is an example of a random time change representation in which stochastic equations involve random time changes of other Markov processes (for details, see \cite[Chapter 6]{EK86}). The generator of the the Markov process $X$ is given by
\begin{align*}
\SC{A}f(x)  = \sum_{k=1}^{R} a_k(x) (f(x+\nu_k)-f(x)),
\end{align*} 
that is, for a bounded-measurable function $f$ the quantity 
$$f(X(t)) - f(X(0)) - \int_0^t \SC{A}f(X(s))$$
is a martingale (with respect to the filtration $\{\SC{G}_t\}$ defined in \eqref{G-filt}). Consequently, taking $f(y) = 1_{\{x\}}(y)$, it follows that 
the probability mass function of $X(t)$ satisfies the following
Kolmogorov forward equation (or the master equation in chemical literature)
\begin{align*}
\f{\partial p(x,t)}{\partial t} = \sum_{k=1}^\Nreact [a_k(x-\nu_k)p(x-\nu_k,t) - a_k(x) p(x,t)],
\end{align*}
where $p(x,t) = \PP(X(t)=x).$ Following \cite{KK13}, we next introduce an appropriate scaled process $X^N$ 
which will be a primary object in our error analysis. 
In a typical multi-scale model, the abundance of various species in the reaction system 
can vary over different orders of magnitude. 
Let $\alpha_i >0$ and define $\bar{X}^N_i = X_i/N^{\alpha_i}$. 
The $\alpha_i$ are chosen such that $\bar{X}^N_i = O(1)$; in other words,
$\alpha_i$ measures the order magnitude in abundance for species $S_i$. 
In a typical reaction system, the stochastic rate constants $c_k$ 
can also vary over different orders of magnitude. Therefore, with the same spirit 
we define $d_k = c_k/N^{\beta_k}$ such that $d_k = O(1)$.

Under the above scaling of the state vector and the rate constants, the propensities $a_k$ scale as
$a_k(X) = N^{\beta_{k}+\nu_{k}^{-}\cdot \alpha}\l_{k}(\bar{X}^N)$, where $\alpha=(\alpha_1,\hdots,\alpha_{\Nsp}).$ 
For example, for unimolecular reactions $S_1\rt *$  we get 
$$a_{1}(X) = c_{1}X_{1} = N^{\beta_{1}}d_1N^{\alpha_1}\bar{X}^N_1 = N^{\beta_1+\alpha_1} d_1\bar{X}^N_1=N^{\beta_1+\alpha_1}\l_1(\bar{X}^N), $$
while for bimolecular reactions of the type $S_1+S_2\rt *$ one obtains 
$$a_2(X) = c_2 X_1X_2 = N^{\beta_2+\alpha_1+\alpha_2}d_2\bar{X}^N_1\bar{X}^N_2 = N^{\beta_2+\alpha_1+\alpha_2}\l_2(\bar{X}^N).$$
Note that with these choices of exponents, the functions $\l_k(\cdot)$ are $O(1)$. Oftentimes, it is beneficial to scale time  as well by $t\rt tN^\gamma$. With all of the above scalings, we look at the process $X^N$ defined by $X^N(t) = \bar{X}^N(tN^\gamma)$.
It readily follows from \eqref{RTM0} that $X^N$ satisfies
\begin{align}
\label{RTM1}
X^N(t) &= X^N(0) + \sum_{k=1}^\Nreact \xi_k(N^{\rho_k} \int_0^t \l_k(X^N(s)) ds ) \ \nu^N_k,
\end{align}
where  $\rho_k = \gamma+\beta_k+\nu_{k}^{-}\cdot \alpha$ and $\nu^N_{ki} = \nu_{ki}/N^{\alpha_i}.$

\subsection{Mathematical preliminaries}

We start with the following useful lemma.
\begin{lemma}\label{lemma0}
Let $\xi$ be a unit Poisson process adapted to a filtration $\{\SC{H}_t\} $ and $\s_1, \s_2$ bounded $\{\SC{H}_t\}$-stopping times. Then
\begin{align*}
\EE\le[|\xi(s_1)  -\xi(\s_2)|\ri] = \EE\le[|\s_1 - \s_2|\ri].
\end{align*} 
\begin{proof}
First note that both $\s_1\vee\s_2$ and $\s_1\wedge\s_2$ are $\{\SC{H}_t\}$-stopping times. Since $\xi(t)-t$ is a $\{\SC{H}_t\}$-martingale, by optional sampling theorem 
we have
\begin{align*}
\EE[\xi(\s_1\wedge\s_2)] =\EE[\s_1\wedge\s_2], \quad \EE[\xi(\s_1\vee\s_2)]=\EE[\s_1\vee\s_2].
\end{align*}
The assertion now follows because
\begin{align*}
\EE\le[|\xi(\s_1)  -\xi(\s_2)|\ri] & = \EE\le[\xi(\s_1)\vee \xi(\s_2) - \xi(\s_1)\wedge \xi(\s_2) \ri]\\
& = \EE\le[\xi(\s_1\vee\s_2)\ri]  - \EE\le[\xi(\s_1\wedge\s_2)\ri] =  \EE\le[\s_1\vee\s_2\ri]  - \EE\le[\s_1\wedge\s_2\ri]\\
& = \EE\le[|\s_1  -\s_2|\ri]. 
\end{align*}
Here the second equality holds because $\xi$ is an increasing process.
\end{proof}
\end{lemma}

Now let  $\xi_k, k=1,\hdots,R$ be independent unit Poisson processes and define the filtration
\begin{align*}
\SC{F}'_{\~u}\equiv \s\{\xi_k(s_k): 0\leq s_k\leq u_k, k=1,\hdots, R\},
\end{align*}
where $\~u=(u_1,u_2\hdots,u_R)$ is a multi-index. Let $\SC{F}_{\~u}$ be the completion of the filtration of $\SC{F}'_{\~u}$. With $X$ as in \eqref{RTM0}, define 
$$\tau_k(t)=\int_0^t a_k(X(s))\ ds.$$
Then $\tau(t) = (\tau_1(t),\hdots, \tau_R(t))$ is a multi-parameter $\{\SC{F}_{\tilde u}\}$-stopping time (see \cite[Chapter 6]{EK86}). Consequently,  for two intensity functions $a^1_k, a^2_k$ and the corresponding processes $X^1, X^2$, the following is an outcome of Lemma \ref{lemma0}:
\begin{align}
\non
\EE\le[\le|\xi_k(\int_0^t a^1_k(X^1(s))\ ds)-|\xi_k(\int_0^t a^2_k(X^2(s))\ ds)\ri|\ri] \\
\label{fact0}
= \EE\le[\le|\int_0^t a^1_k(X^1(s))\ ds)-\int_0^t a^2_k(X^2(s))\ ds\ri|\ri].
\end{align}

Next define the filtration $\{\SC{G}_t\}$ by
\begin{align}\label{G-filt}
\SC{G}_t = \SC{F}_{\tau(t)}.
\end{align}
Then notice that by the optional sampling theorem, for each $k$, $\tilde{\xi}_k(\tau_k(t)) = \xi_k(\tau_k(t)) - \tau_k(t)$ is a $\{\SC{G}_t\}$-martingale.

\subsection{Hybrid Diffusion Models}
\label{ssec:sde}
 \vs{.2cm} 
 After a possible renaming of the species, assume that reaction \(R_1\) is of the type $S_1+ S_2 \rt S_3$.
The goal of this section is to compute a bound for the error when the reaction $R_1$ is simulated according to a diffusion approximation. At the process level, this typically means that we are replacing the process $\xi_1(t)$ with $W_1(t)+t$, where $W_1$ is a standard Brownian motion. With this change, the approximating process $Z^N$ satisfies the equation
\begin{align*}
Z^N(t) &= X^N(0) + N^{\rho_1} \int_0^t \l_1(Z^N(s) )ds\:\nu^N_1 + W_1(N^{\rho_1} \int_0^t \l_1(Z^N(s) )ds)\nu^N_1\\
&\hspace{.4cm}+\sum_{k>1} \xi_k(N^{\rho_k} \int_0^t \l_k(Z^N(s)) ds  )\ \nu^N_k.
\end{align*}
The goal of this section is to bound the error $e(t) = \EE|X^N(t)-Z^N(t)|$. It should be noted that the error $e$ depends on the coupling between the processes $X^{N}$ and $Z^{N}$. In particular, this means that $e$ will depend on the construction of the Brownian motion $W_1$. The following lemma proves the existence of a Brownian motion $W_1$ on the same probability space as $\xi_1$ (see \cite[Chapter 11, Section 3]{EK86}).
{
\begin{lemma}\label{BMconst}
There exists a Brownian motion $W_1$ on the same probability space as $\xi_1$ such that $W_1$ is independent of $\xi_k$, $k\neq 1$
\begin{align*}
\sup_t \f{|\tilde{\xi}_1(t) - W_1(t)|}{\log(2\vee t)}<\infty,
\end{align*}
where $\tilde{\xi}_1(t)  = \xi_1(t)-t$ denotes the centered Poisson process. Furthermore, for $\delta,\kappa>0$, there exist constants $\theta,K,C>0$ such that
\begin{align*}
\PP\le[\sup_{t\leq \delta n}|\tilde{\xi}_1(t) - W_1(t)|>C\log n+x\ri] \leq Kn^{-\kappa}e^{-\theta x}.
\end{align*}
\end{lemma}

\np
Let 
\begin{align}\label{Rkdef}
\SC{R}_k&=\{i: \nu_{ik}\neq 0\}.
\end{align}
Notice that for each $k$, $\SC{R}_k$ keeps track of the species involved in reaction $R_k$.
We are now ready to state our main result.

\begin{theorem}
\label{thm:diffussion}
Let $X^N$ be given by \eqref{RTM1} and $Z^N$ by
\begin{align*}
Z^N(t) &= X^N(0) + N^{\rho_1} \int_0^t \l_1(Z^N(s) )ds\:\nu^N_1 + W_1(N^{\rho_1} \int_0^t \l_1(Z^N(s) )ds)\nu^N_1\\
&\hspace{.4cm}+\sum_{k>1} \xi_k(N^{\rho_k} \int_0^t \l_k(Z^N(s)) ds  )\ \nu^N_k,
\end{align*}
where $W_1$ is a standard Brownian motion independent of the $\xi_k$ as given by Lemma \ref{BMconst}.
Assume that the $\l_k$ are Lipschitz continuous 
with Lipschitz constant $L_k$ and  $\displaystyle \sup_x \l_k(x) \leq \bar{\l}_k$.
Let $|\nu^N_k| =O(N^{-m_k}) = O(\displaystyle \sum_{i\in \SC{R}_k}\f{1}{N^{\alpha_i}}).$ Then,
\begin{equation}
\label{arnab_error_bound1}
\sup_{t\leq T}\EE|X^N(t) - Z^N(t)| \leq C_N (C'\log N^{\rho_{1}}/N^{m_1} + K''/N^{2\rho_1+m_1}),
\end{equation}
where $C_N$ is a constant  which remains the same no matter which reaction is simulated by a diffusion approximation.
\end{theorem}
\begin{proof}
Notice that
for $i=1,2,3$
\begin{align*}
Z^N_i(t) &= X^N_i(0) + N^{\rho_1} \int_0^t \l_1(Z^N(s)) ds\:\nu^N_{1i} + W_1(N^{\rho_1} \int_0^t \l_1(Z^N(s) )ds)\nu^N_{1i}\\
& \hspace{.4cm}+\sum_{k>1} \xi_k(N^{\rho_k} \int_0^t \l_k(Z^N(s))ds ) \ \nu^N_{ki},
\end{align*}
and for $i>3$
\begin{align*}
Z^N_i(t) = X^N_i(0) + \sum_{k>1} \xi_k(N^{\rho_k} \int_0^t \l_k(Z^N(s) )ds )\ \nu^N_{ki}.
\end{align*}
Let $\tilde{\xi}$ denote the centered Poisson process. Observe that for $i=1,2,3$,
\begin{align}
\non
|X^N_i(t) -Z^N_i(t)| &\leq N^{\rho_1}|\int_0^t (\l_1(X^N(s)) - \l_1(Z^N(s)) )\ ds| |\nu^N_{1i}|\\
\non
&\hspace{.4cm} + | \tilde{\xi}_1(N^{\rho_1} \int_0^t \l_1(X^N(s) )ds ) -  W_1(N^{\rho_1} \int_0^t \l_1(Z^N(s) )ds)|\ |\nu^N_{1i}|\\
\non
&\hspace{.4cm}+ \sum_{k>1} |\xi_k(N^{\rho_k} \int_0^t \l_k(X^N(s)) ds ) -\xi_k(N^{\rho_k} \int_0^t \l_k(Z^N(s)) ds )  |\ |\nu^N_{ki}|\\
& = A+B+C.
\label{react_changed}
\end{align}
Note that by \eqref{fact0}
\begin{align*}
\EE[C] &\leq \sum_{k>1}N^{\rho_k} |\nu^N_{ki}| \EE |\int_0^t (\l_k(X^N(s))-\l_k(Z^N(s)) )\ ds|\\
&\leq \sum_{k>1}N^{\rho_k} |\nu^N_{ki}| L_k \int_0^t \EE|X^N(s) - Z^N(s)| ds,
\end{align*}
where $L_k$ is the Lipschitz constant for $\l_k$.
Also, it is immediate that 
$$\EE[A] \leq N^{\rho_1}|\nu^N_{1i}|L_1\int_0^t \EE|X^N(s) - Z^N(s)| ds.$$
Next, observe that
\begin{align*}
B&\leq  | \tilde{\xi}_1(N^{\rho_1} \int_0^t \l_1(X^N(s) )ds) -  \tilde{\xi}_1(N^{\rho_1} \int_0^t \l_1(Z^N(s)) ds)|\ |\nu^N_{1i}|\\
&\hspace{.4cm}+  | \tilde{\xi}_1(N^{\rho_1} \int_0^t \l_1(Z^N(s) )ds ) -  W_1(N^{\rho_1} \int_0^t \l_1(Z^N(s))ds)|\ |\nu^N_{1i}|\\
& = I + II.
\end{align*}
It is easy to see that for some constant $C$,
$$\EE[I] \leq C N^{\rho_1}|\nu^N_{1i}|L_1\int_0^t \EE|X^N(s) - Z^N(s)| ds.$$
By Lemma \ref{BMconst}, there exist constants $\gamma, K', C'>0$ such that 
\begin{align*}
II/|\nu^N_{1,i}|&\leq \sup_{s\leq\bar{\l}_1 N^{\rho_1}t} |\tilde{\xi}_1(s)-W_1(s)|\leq C'\log N^{\rho_1} + \Lambda_N,
\end{align*}
where $\PP[\Lambda_N >x]\leq K'e^{-\kappa x}/N^{2\rho_1}$. Notice that 
$$\EE[\Lambda_N] = \int_0^\infty \PP[\Lambda_N>x]\  dx \leq K'/\kappa  N^{2\rho_1}.$$
It follows that
\begin{align*}
\EE[II] \leq( C'\log N^{\rho_1} + K''/N^{2\rho_1})|\nu^N_{1i}|.
\end{align*}

For $i>3$, we have
\begin{align}
\non
\EE|X^N_i(t) -Z^N_i(t)| &\leq \sum_{k>1}\EE |\xi_k(N^{\rho_k} \int_0^t \l_k(X^N(s)) ds ) -\xi_k(N^{\rho_k} \int_0^t \l_k(Z^N(s)) ds )  |\ |\nu^N_{ki}|\\
& \leq \sum_{k>1}N^{\rho_k}|\nu^N_{ki}|\EE\int_0^t |\l_k(X^N(s)) - \l_k(Z^N(s))| ds.
\label{react_unchanged}
\end{align}

\np
It follows  from \eqref{react_changed} and \eqref{react_unchanged} that after summing over $i$
\begin{align*}
\EE|X^N(t) - Z^N(t)| &\leq  (\displaystyle\sum_{k=1}^{R}N^{\rho_k} |\nu^N_{k}| L_k+N^{\rho_1}L_1|\nu^N_{1}|C )\int_0^t \EE|X^N(s) - Z^N(s)| ds \\
&\hspace{.4cm} + ( C'\log N^{\rho_1} + K''/N^{2\rho_1})|\nu^N_{1}|,
\end{align*}
where $|\nu^N_k| = \displaystyle\sum_i |\nu^N_{ki}|$ is the $l_1$ norm of $\nu^N_k$. Since $|\nu^N_{k}| = O(1/N^{m_k})$,  Gronwall's inequality implies 
\begin{equation}
\label{error:sde}
\EE|X^N(t) - Z^N(t)| \leq C_N (C'\log N^{\rho_1}/N^{m_1} + K''/N^{2\rho_1+m_1}),
\end{equation}
where $C_N =\exp(\displaystyle\sum_{k=1}^{R}N^{\rho_k} |\nu^N_{k}| L_kt+N^{\rho_1}L_1|\nu^N_{1}|C t) \leq \exp(\displaystyle 2\sum_{k=1}^{R}N^{\rho_k} |\nu^N_{k}| L_kt).$ 
\end{proof}

\begin{remark} {\rm
 Typically, for biochemical systems the propensities $a_k$ and hence the $\l_k$ will not be bounded, a condition required in Theorem \ref{thm:diffussion}. However, for most practical purposes our simulation takes place in a bounded domain, that is, the simulation is stopped if the number of molecules exceed a certain quantity. Hence, the assumption that the propensity functions are bounded remains valid. Specifically, one way to incorporate this feature in our model is to multiply the original propensity function by a cutoff function ensuring that the changed propensity function vanishes outside a bounded region.
 }
\end{remark}

\begin{remark}{\label{remark1}\rm
For simulation purposes, it is often difficult to estimate the strong error $\EE|X^N(t)  -Z^N(t)|$, as it requires proper utilization of the coupling between $X^N$ and $Z^N$. It is often more convenient to look at a weak error which compares the error between the marginal distribution of $X^N$ and that of $Z^N$ at a time $t$. For example, depending on the need, a practitioner might want  to estimate the weak error given by $|\EE(X^N(t)) - \EE(Z^N(t))|$. While this weak error compares the average values of the exact and the approximating processes, it does not shed light on the error at the level of the corresponding probability distributions. This can be accurately captured by the Wasserstein distance
\begin{align*}
d_W(X^N(t),Z^N(t)) &= \sup_{f}\big\{|\EE(f(X^N(t)) -\EE(f(Z^N(t))|: f:\R_+^M\rt \R_+, \\
&\hs{2cm}  Lip(f)\leq 1\big\}.
\end{align*}
Here, the supremum is taken over all Lipschitz continuous $f$, and $Lip(f)$ denotes the corresponding Lipschitz constant.  It is immediate that
$$d_W(X^N(t),Z^N(t)) \leq\EE|X^N(t)  -Z^N(t)|,$$
and the latter quantity can be bounded by the error bound obtained in Theorem \ref{thm:diffussion}.

}
\end{remark}

\section{Simulation Algorithms}
\label{sec:algorithm}
 \vs{.2cm} 
 The goal of this section is to use the obtained error bound for constructing a fast  algorithm for the dynamic partitioning of the reaction set. 
 The objective of such an algorithm is to implement a proper protocol for simulating the fast reactions by diffusion approximations and to switch back to the original exact Gillespie simulation when the conditions for  approximation are not met. Furthermore, switching back and forth between exact simulation and diffusion approximation of appropriate reactions will be done dynamically over the course of time.

The error bound in Section 
\ref{sec:model}, was calculated under the assumption that the system consists of a single fast reaction which was then approximated by a diffusion approximation. This can easily be generalized to a system consisting of more than one fast reaction. Analysis of the error bound in \eqref{arnab_error_bound1} reveals that it consists of products of two terms, the first term being a constant (in the sense that its value does not depend on which reaction is approximated by diffusion term), while the second term explicitly captures the effect  of the specific reaction being approximated. Consequently, it is the second part of this error bound which is utilized in partitioning the reaction set into the fast and the slow reaction set. The specific mathematical details are outlined below.

It should be noted that the scaling constants $N$ and the exponents $\rho_k, m_k$, appearing in \eqref{arnab_error_bound1}, are determined based on the starting initial state of the system. Hence,  a key step in the  development of an effective algorithm involves  rewriting (\ref{arnab_error_bound1}) in terms 
of the propensity functions and the number of molecules of species.

Assume that the reaction \(R_k\) is simulated 
by a diffusion approximation and define  
\begin{equation}
\label{error_sde1}
\Upsilon_k= C'\log N^{\rho_{k}}/N^{m_k} + K''/N^{2\rho_k+m_k}.
\end{equation}
Recall that
\[a_k(X)=N^{\rho_k}\lambda_k(X^{N}),\]
where \(\lambda_k(.)=O\left(1\right)\). Therefore,  
\begin{equation}
\label{propensity1}
a_k(X)=O(N^{\rho_k}) .
\end{equation} 
Also,  \(|\nu^N_k| =O(N^{-m_k}) = O(\displaystyle\sum_{i\in \SC{R}_k}    \f{1}{N^{\alpha_i}} )\).
Since, $\bar{X}^N_i = X_i/N^{\alpha_i}$ where  
$\bar{X}^N_i = O(1)$, we have 
\(\displaystyle\f{1}{X_i}=O(\displaystyle\f{1}{N^{\alpha_i}} )\)
and
\begin{equation}
\label{propensity2}
 O(N^{-m_k}) =O(\displaystyle \sum_{i\in R_k } \f{1}{X_i}) .
\end{equation} 
Ignoring the   constants \(C'\), \(K''\) and using
(\ref{error_sde1}), (\ref{propensity1}) and (\ref{propensity2}) 
it follows that the effect of simulating the reaction $R_k$ by a diffusion approximation is essentially captured by  
\begin{equation}
\label{error_bound1}
\hat{\Upsilon}_k\equiv \displaystyle\sum_{i\in \SC{R}_k }\f{\log a_k(X)}{X_i} +\f{1}{a^{2}_k(X)X_i},
\end{equation}
where $\SC{R}_k$ was defined in \eqref{Rkdef}.  
Now given a threshold $\varepsilon$, a reaction $R_k$ is classified as a fast reaction and simulated by diffusion approximation if $\hat{\Upsilon}_k\leq \varepsilon$. $\hat{\Upsilon}_k$ can be considered as an effective estimate of the quantity $\Upsilon_k$ calculated using the  ``initial condition'' of the state. This approach, in particular, introduces a systematic way of partitioning the reaction set into fast and slow reactions based on a user defined threshold over the course of time. The resulting algorithm is outlined in details in Algorithm \ref{sumalgo}.

\begin{algorithm}[H]
\DontPrintSemicolon
\KwIn{The state vector \(X\),  the error bound \( \varepsilon\), 
a discretization time step \(\Delta\), stoichiometric matrix  \(S=\{\nu_{ij}\}\), \(i=1,2,\ldots, \Nsp,\:j=1,2,\ldots,\Nreact\), 
a positive number \( P \)  to check repartitioning of reactions, end of the simulation time 
 \( T >0\). }
\KwOut{ The number of molecules of each species in time interval \( t \in [0,T] \). }

Set \(t=0,n=0\).\;
Calculate  \(\hat{\Upsilon}_k\) for all reactions  by using  equation (\ref{error_bound1}).\; 
Partition the reaction set into \(\mathcal{C}\) (continuous) and \(\mathcal{D}\)  (discrete) sets such that for each  \( R_k\in \mathcal{C}\),
\(\hat{\Upsilon}_k \leq \varepsilon\) and  for each  \(R_m\in \mathcal{D}\),
\(\hat{\Upsilon}_m > \varepsilon\) \label{partition}.\;
For each \(R_m\in \mathcal{D}\), set  \(T_{m}=0\).\;
For each \(R_m\in \mathcal{D}\), draw \(J_{m}\sim -\log(z),\quad z\sim \mathcal{U}(0,1)\).\;

\While{  \(t<T\) \textbf{ and } \(\displaystyle \sum_{i=1}^{R} a_{i} >0\)} {
     \( n=n+1\).\;
  
     For each \(R_m\in \mathcal{D}\), calculate \(h_m= \frac{(J_m - T_m)}{a_m}\).\;
     Choose \(\alpha\) such that \(h\equiv h_{\alpha}\equiv \displaystyle \min_{R_m\in \mathcal{D}} h_m\).\;
 
   \uIf { \(\Delta>h\) }
      {Update  \(X\) by a suitable numerical scheme for simulating the Langevin dynamics for \(R_k \in \mathcal{C} \) until the $h$.\;
        Carry out reaction \(R_\alpha\)  and update \(X=X+\nu_\alpha\).\;
        Update $J_\alpha = J_\alpha-\log(u)$, $u\sim \mathcal{U}(0,1).$\\
        For each $R_m\in \SC{D}$, put $T_{m}=T_{m}+a_{m}h.$\\
        \(t=t+h\).}
        \Else
        {Update  \(X\) by a suitable numerical scheme for simulating the Langevin dynamics for \(R_k \in \mathcal{C} \) until the $\Delta$.\;
       For each $R_m\in \SC{D}$, \(T_{m}=T_{m}+a_{m}\Delta\).\;
        \(t=t+\Delta\).\;
    }
Recalculate the propensities of all reactions \(a_{k}\). \;
\If{ \( n\equiv 0 \pmod{P}\)}{
Recalculate errors of all reactions \(\hat{\Upsilon}_k\).\;
Repartition reactions  as in Step \ref{partition}.
}
}
\caption{ Dynamic partitioning algorithm for hybrid diffusion model. }
\label{sumalgo}
\end{algorithm}

\section{Numerical Scheme for Hybrid Diffusion Models}
\label{sec:runge_kutta}
  \vs{.2cm}
  Suppose that the reaction set is partitioned into a set of fast reactions $\SC{C}$ and a set of slow reactions $\SC{D}$. The reactions in $\SC{C}$ will be simulated by the diffusion approximation. The resulting approximating state process is given by
\begin{equation}
\begin{array}{rcl}
\label{eq:1}
X(t)&=& X(0)+\displaystyle \sum_{\ell \in \mathcal{C}}  \int_{0}^{t} a_{\ell}(X(s))ds \nu_{\ell} +\displaystyle \sum_{\ell \in \mathcal{C}} W_{\ell}(\int_{0}^{t} a_{\ell}(X(s))ds)\nu_{\ell}\\
&+&\displaystyle \sum_{k \in \mathcal{D}} \xi_{k} (\int_{0}^{t} a_{k}(X(s))ds)\nu_{k}.
\end{array}
\end{equation}

Therefore, the evolution of the process is governed by a usual diffusion process punctuated by jumps from the slow reaction set. Specifically, let $\tau_1$ and $\tau_2$ denote two successive of reactions from $\SC{D}$. Then, for $\tau_1<t<\tau_2$,
\begin{equation}
\label{eq:2}
X(t)= X(\tau_1)+\displaystyle \sum_{\ell \in \mathcal{C}}  \int_{\tau_1}^{t} a_{\ell}(X(s))ds \nu_{\ell}+\displaystyle \sum_{\ell \in \mathcal{C}} W_{\ell}(\int_{\tau_1}^{t} a_{\ell}(X(s))ds)\nu_{\ell}.\\
\end{equation}
Since the reactants in fast reactions usually involve species with high copy numbers, it is useful to look at their  concentration $U(t)$ defined by \(U(t)=\Omega^{-1}X(t)\), 
where \(\Omega\) is the volume of the reaction compartment and $\tau_1<t<\tau_2$.

Notice that the propensity 
function of the reaction \(R_k\) satisfies  
the following relation 
\begin{equation}
\label{eq:3}
a_{k}(X)\approx \Omega \: \tilde{a}_{k}(U), 
\end{equation} 
where \(\tilde{a}_{k}: \R_{\geq 0}^{M} \rightarrow \R \) 
is the usual deterministic form of mass action \(\tilde{a}_{k}(U)=\tilde{c}_k \displaystyle \prod_{i=1}^\Nsp U_{i}^{ \nu_{ik}}\) where \(\tilde{c}_k\) denotes the deterministic rate constant. 
It should be noted that the above relation is exact for unimolecular reactions and bimolecular reactions of the type $S_1+S_2\rt *$, and for reactions of the type $2S_1\rt *$, the error is of the order $O(\Omega^{-1})$. Consequently, for $\tau_1<t<\tau_2$, $U$ satisfies
\begin{equation}
\label{eq:5}
U(t)= U(\tau_1)+\displaystyle \sum_{\ell \in \mathcal{C}} \int_{\tau_1}^{t} \tilde{a}_{\ell}(U(s))ds\nu_{\ell} +\displaystyle \sum_{\ell \in \mathcal{C}} \frac{1}{\Omega}W_{\ell}(\Omega \int_{\tau_1}^{t} \tilde{a}_{\ell}(U(s))ds)\nu_{\ell}, \\
\end{equation}
which is equivalent (in the sense of distribution) to the equation
\begin{equation}
\label{eq:6}
U(t)= U(\tau_1)+\displaystyle \sum_{\ell \in \mathcal{C}}  \int_{\tau_1}^{t} \tilde{a}_{\ell}(U(s))ds \nu_{\ell} +\displaystyle \sum_{\ell \in \mathcal{C}}\frac{1}{\sqrt{\Omega}}\int_{\tau_1}^{t}\sqrt{\tilde{a}_{\ell}(U(s))}dB_{\ell}(s)\nu_{\ell}.\\
\end{equation}
Here, the \(B_{\ell}\) are independent standard Brownian motions.  Assume that the set $\SC{C}$ has \(C\) reactions. Then, notice that
\begin{equation}
\label{eq:7}
dU_{i}=f_{i}(U)dt+\frac{1}{\sqrt{\Omega}}\displaystyle \sum_{j=1}^{C} g_{ij}(U)dB_{j},\quad i \in \SC{S}_{\SC{C}},
\end{equation}
where $\SC{S}_{\SC{C}}$ denotes the species involved in $\SC{C}$,
\begin{equation}
\label{eq:8}
f_{i}(U)=\sum_{j=1}^{C}\nu_{ij}\:\tilde{a}_{j}(U),\quad g_{ij}(U)=\nu_{ij}\sqrt{\tilde{a}_{j}(U)}.
\end{equation}

A trajectory of the above stochastic differential equation is simulated by using the Runge-Kutta method 
with strong  order \(2\)  as proposed in \cite{bb:96}.  The first step involves rewriting (\ref{eq:7}) in the equivalent
Stratonovich form
\begin{equation}
\label{eq:9}
\mathrm{d}U_{i}=\bar{f}_{i}(U)\mathrm{dt}+\frac{1}{\sqrt{\Omega}}\displaystyle \sum_{j=1}^{C}g_{ij}(U) \circ \mathrm{d}B_{j},
\end{equation}
where
\begin{equation}
\label{eq:10}
\bar{f}_{i}(U)=f_{i}(U)- \frac{1}{2\sqrt{\Omega}}\displaystyle \sum_{j=1}^{M}\sum_{k=1}^{C}g_{jk}(U) 
\frac{\partial g_{ik}(U)}{\partial U_{j}},
\end{equation}
with the symbol \(\circ\) denoting the Stratonovich integral. The SDE  can be represented in matrix form as
\begin{equation}
\label{eq:11}
\mathrm{d}U=(S\tilde{a}(U)- \frac{1}{2 \sqrt{\Omega}} S h(U)) dt+ \frac{1}{\sqrt{\Omega}} S\gamma(U) \circ \mathrm{d}B,
\end{equation}
where \( \tilde{a}(U)=( \tilde{a}_{1}(U),\tilde{a}_{2}(U),\ldots,\tilde{a}_{C}(U) )^{T} \), \( B=( B_{1},B_{2},\ldots,B_{C} )^{T} \),  \(\gamma(U)=\mathrm{diag}(\sqrt{\tilde{a}_{1}(U)},\ldots,\sqrt{\tilde{a}_{C}(U)})\), 
and the entries of the vector \(h(U)\) are given by
\[h_{k}(U)=\displaystyle \sum_{j=1}^{M} \frac{\partial \, \tilde{a}_{k}(U)}{\partial U_{j}}\nu_{jk},\:\: k=1,2,\ldots,C.\]
Notice that  
for unimoleculer reactions \(S_{k} \rt *\), we have
\[\ \frac{\partial \, \tilde{a}_{k}(U)}{\partial U_{j}}=\begin{cases} \tilde{c}_{k} & \mbox{ if } j=k\\
0 & \mbox{ if } j \neq k,\end{cases}\]
and for bimolecular reactions \(S_{k}+S_{i} \rt *\), we obtain 
\[ \frac{\partial \,  \tilde{a}_{k}(U)}{\partial U_{j}}=\begin{cases} \tilde{c}_{k}U_{k} & \mbox{ if } j=i\\
\tilde{c}_{k}U_{i} & \mbox{ if } j=k.\end{cases}\]
Given \(U(\tau_{1})\) as the approximate solution at time \(\tau_{1} \),
the four stages explicit Runge-Kutta method with strong order \(2\)
for the Stratonovich problem  (\ref{eq:11}) gives  the following intermediate values 
\(I_{s}, s=1,2,3,4,\)
\begin{equation}
\label{eq:30}
\begin{array}{rcl}
I_{s}&=&U(\tau_{1})+h\displaystyle \sum_{j=1}^{s-1} A_{sj}( S\tilde{a}(I_{j})-\frac{1}{2 \sqrt{\Omega} }S h(I_{j}) )\\
&+&\displaystyle \sum_{j=1}^{s-1} \frac{1}{\sqrt{\Omega}} S \gamma(I_{j}) (B_{sj}^{(1)}J_{1}+B_{sj}^{(2)}\frac{J_{10}}{h}),
\end{array}
\end{equation}
from which the approximate solution at  \(\tau_{1}+h\) is formed: 
\begin{equation}
\label{eq:31}
\begin{array}{rcl}
U(\tau_{1}+h)&=&U(\tau_{1})+h\displaystyle \sum_{j=1}^{4}\mu_{j}(S\tilde{a}(I_{j})-\frac{1}{2\sqrt{\Omega}}S h(I_{j}))\\ &+&\displaystyle \sum_{j=1}^{4} \frac{1}{\sqrt{\Omega}} S \gamma(I_{j}) (\eta_{j}^{(1)} J_{1}+\eta_{j}^{(2)} \frac{J_{10}}{h}).
\end{array}
\end{equation}
Here, \(A=\{A_{sj}\}\), \(B^{(k)}=\{B^{(k)}_{sj}\}\), \(s,j \in \{1,2,3,4\}\) are \(4\times 4\)
matrices of real elements, \(\mu=\{\mu_{j}\}\), 
\(\eta^{(k)}=\{\eta^{(k)}_{j}\}\), 
\( k=1,2,\: j=1,2,3,4\)
are row vectors of real elements and \(J_{1},\:J_{10} \) 
are \(C \times 1\) column vectors of real elements. 
Here, \(J_{1}\) denotes the Wiener increment and 
\(J_{10}\) is an approximation for the  
Stratonovich multiple integral 
\(J=\int \int \circ dB ds\). For details  of the Runge-Kutta method with strong order \(2\) and the expressions for the terms \(A\), \(B^{(k)}\), \(\mu\), \(\eta^{(k)}\), \(J_{1}\),  \(J_{10}\), \(k=1,2\), see \cite{bb:96,kp:92}. 

The prescribed scheme has strong order \(2\) but imposes a fixed stepsize that need to be chosen very small if the system is stiff. For such stiff problems,  we  employ a heuristic to adaptively choose the stepsizes of the Runge-Kutta method
for a certain integration accuracy. In order to estimate the stepsizes of  the Runge-Kutta method for stochastic differential equations (SDEs), we follow the adaptive stepsize control for the classical  Runge-Kutta methods for ODEs. To determine the stepsizes of the Runge-Kutta method with strong order \(2\) for the Stratonovich problem (\ref{eq:11}), we use just the drift term of (\ref{eq:11}).\footnote{Note that in the Stratonovich form the drift term shows a volume dependency.}More specifically, we choose an initial step \(h\) and compute two approximate solutions of  the drift term 
with stepsizes \(h\) and \(h/2\) by using the fourth order classical Runge-Kutta method for ODEs. If the difference of the approximations is smaller than a given tolerance,
then we choose \(h\) as the stepsize  and compute the approximate solution of the whole SDE  given by (\ref{eq:11}) by using the Runge-Kutta method  with strong order \(2\). We refer the reader to  \cite{hnw:93}, 
for more details on the adaptive stepsize algorithms for ODEs.

 \section{Conversion to Differential Algebraic Form}
\label{sec:conservation}
  \vs{.2cm}
Mass conservation relations play an important role in biochemical reaction systems. In many models the reaction dynamics  dictates 
conservation of the total amount of two or more species over the course of time.  For example, in the Michaelis-Menten model  the total quantity of the enzyme and the enzyme-substrate complex is always conserved (see Section \ref{application:MM1}). 
These constraints can be defined by algebraic equations. 
As a result, the dynamics of the systems under consideration can be expressed 
by differential algebraic forms that more generally,  preserve the symplectic structure of the model 
on the constraint manifold \cite{hlw:06,hw:96}. These algebraic relations
leads to reduction in the dimensionality of the equation set, which in turn speeds up the simulation. The procedure is detailed below.

The algebraic constraints indicate that  \(r \equiv\mbox{Rank}\left\{S\right\}<M\), where $M$ is the number of species in the reaction system.  The next step involves converting the stoichiometric matrix $S$ into a reduced echelon form by Gauss-Jordan method (see \cite{ch:02,sl:04} for more ways of reconstructing the equations describing the dynamics of such  a reaction system). Specifically, the method gives a permutation matrix  \(E \in \N^{M\times M}\)  (that is, $E$ is a product of elementary matrices) such that

\begin{equation}
\label{eq:42}
E S=\left[
    \begin{array}{c}
      E_{I} \\
       E_{D}\\
       \end{array}
  \right] S =\left[
    \begin{array}{c}
      S_{I} \\ 
         0 \\
       \end{array}
  \right],
\end{equation}
where \(E_{I}, E_{D} \) are \(r  \times M \)
and \(\left(M- r \right) \times M\) matrices, respectively. 
Notice that  \(S_{I}\equiv E_{I}S\)  has rank $r$. 
Here, \(E_{D}S=0\) means that \(E_{D}\) can be thought of as the conservation matrix. 
By \eqref{eq:11}, we have 
\begin{equation}
\label{eq:49}
E \mathrm{d}U=E (S\tilde{a}(U)- \displaystyle \frac{1}{2 \sqrt{\Omega}} S h(U)) \mathrm{d}t+\frac{1}{\sqrt{\Omega}}E S \gamma (U)\circ \mathrm{d}B, 
\end{equation}
and hence,
\begin{eqnarray*}
\left[
    \begin{array}{c}
       \mathrm{d}U_{I} \\
         \mathrm{d}U_{D}\\
       \end{array}
  \right]&\equiv & \left[\begin{array}{c}
       E_I \mathrm{d}U \\
         E_D \mathrm{d} U\\
       \end{array}
  \right]\\
&=&\left[
    \begin{array}{c}
       E_{I}(S\tilde{a}(U)- \displaystyle \frac{1}{2 \sqrt{\Omega}} S h(U))\mathrm{d}t+\displaystyle \frac{1}{\sqrt{\Omega}} E_{I} S \gamma(U)\circ \mathrm{d}B \\
        E_{D}(S\tilde{a}(U)- \displaystyle \frac{1}{2 \sqrt{\Omega}} S h(U))\mathrm{d}t+\displaystyle \frac{1}{\sqrt{\Omega}} E_{D} S \gamma(U)\circ \mathrm{d}B \\
       \end{array}
  \right]\\
&=&\left[
    \begin{array}{c}
       E_{I}(S\tilde{a}(U)- \displaystyle \frac{1}{2 \sqrt{\Omega}} S h(U))\mathrm{d}t+\displaystyle \frac{1}{\sqrt{\Omega}} E_{I} S \gamma(U)\circ \mathrm{d}B \\ \\
         0 \\
       \end{array}
  \right].
\end{eqnarray*}

It follows that \eqref{eq:11} can be reduced to
\begin{subequations}
\begin{eqnarray}
\mathrm{d}U_{I}&=&E_{I}(S\tilde{a}(U)- \frac{1}{2 \sqrt{\Omega}} S h(U))  \mathrm{d}t+\displaystyle \frac{1}{\sqrt{\Omega}} E_{I} S \gamma(U)\circ \mathrm{d}B  \label{eq:511}\\
U_{D} &=&E_{D}U=C\label{eq:512},
\end{eqnarray}
where $C$ is a constant with respect to time.
\end{subequations}
Now writing
\begin{align*}
E^{-1}&=\left[
    \begin{array}{c}
       \Psi \\ 
         \Phi\\
       \end{array}
  \right],
\end{align*}
we have $U = \Psi U_I+\Phi U_D = \Psi U_I+\Phi C$, where the last equality is because of \eqref{eq:512}. Consequently, $U_I$ satisfies
\begin{align*}
\mathrm{d} U_{I}&=E_{I}(S\tilde{a}(\Psi U_I+\Phi C)- \frac{1}{2 \sqrt{\Omega}} S h(\Psi U_I+\Phi C))  \mathrm{d} t\\
&\hs{0.5cm}+\frac{1}{\sqrt{\Omega}} E_{I} S \gamma(\Psi U_I+\Phi C)\circ \mathrm{d}B,
\end{align*}
and takes its values in a lower dimensional space compared to the original process $U$. The trajectories of $U_I$ can now be simulated by the Runge-Kutta method as outlined in Section \ref{sec:runge_kutta}. Specifically, given \(U_{I}(\tau_{1})\),  \eqref{eq:30} gives the following intermediate values  for the independent variables  
\begin{equation*}
\begin{array}{rcl}
I_{s}&=&U_{I}(\tau_{1})+h\displaystyle \sum_{j=1}^{s-1} A_{sj}E_{I} ( S\tilde{a}(R_{j})-\frac{1}{2 \sqrt{\Omega}} S h(R_{j}))\\
&+&\displaystyle \sum_{j=1}^{s-1} 
 \frac{1}{\sqrt{\Omega}} E_{I}  S \gamma(R_{j})(B_{sj}^{(1)} J_{1}+B_{sj}^{(2)}\frac{J_{10}}{h}),
\end{array}
\end{equation*}
where \(R_{s}=\Psi I_{s}+\Phi C,\:\: s=1,2,3,4\). Finally, we obtain the approximate values of the independent variables at time 
\(\tau_{1}+h\) as follows
\begin{equation}
\begin{array}{rcl}\label{last}
U_{I}(\tau_{1}+h)&=&U_{I}(\tau_{1})+h\displaystyle \sum_{j=1}^{4}\mu_{j} E_{I} (S\tilde{a}(R_{j})-\frac{1}{2 \sqrt{\Omega}}  S h(R_{j}))\\
&+&\displaystyle \sum_{j=1}^{4} \frac{1}{\sqrt{\Omega}} E_{I}  S \gamma(R_{j}) (\eta_{j}^{(1)} J_{1}+\eta_{j}^{(2)} \frac{J_{10}}{h}).
\end{array}
\end{equation}
Also, $U(\tau_{1}+h)$ can be easily obtained from the following equation:
\begin{equation}
\label{last_state}
U (\tau_{1}+h)=\Psi U_{I}(\tau_{1}+h)+\Phi C.
\end{equation}

\section{Applications}
\label{sec:application}
 \vs{.2cm}
In this section, the proposed algorithm from Section \ref{sec:algorithm} 
is applied to  the Michaelis-Menten kinetics, the Lotka-Volterra model and a large-scale MAPK pathway model  
together with its gene expression. The validity of the obtained theoretical error bound for the Michaelis-Menten model is substantiated empirically in Section  \ref{application:MM2}.
The enormous advantage of our hybrid algorithm over exact stochastic simulation in terms of computational efficiency will be demonstrated
in Section  \ref{application:MAPK}  by  considering  the complex MAPK pathway.

\subsection{The Michaelis-Menten Model }
\label{application:MM1}
  
\vs{.2cm}
The well known Michaelis-Menten model for enzymatic  substrate conversion consists of four species,  the enzyme (E), the substrate (S), 
the enzyme-substrate complex (ES) and the product (P). 
These 
species interact via the following reaction channels 
\begin{equation}
\label{MM_reactions}
E+ S \xrightleftharpoons[d_{1}]{c_{1}} ES, \quad ES\stackrel{c_{2}}{\longrightarrow}  E+P. 
\end{equation}
The state of the system is defined by the vector of copy numbers \(X=(E,S,ES,P)^{T}\). Notice that the following conservation laws hold
\begin{align*}
E+ES = C_1, \quad S+ES+P=C_2.
\end{align*}
Here, $C_1$ and $C_2$ are constants (with respect to time) and will be considered as dependent variables.
In our numerical simulation study, the initial 
number of molecules  is  taken as \(X(0)=\left(48,298,2,0\right)^{T}\)
and the rate constants  of  reactions \(R_1\), \(R_2\), \(R_3\) are given by \(c_{1}=0.02\) \(\mathrm{molec^{-1}s^{-1}}\), \(d_{1}=0.5\) \(\mathrm{s^{-1}}\) , \(c_{2}=0.1\) \(\mathrm{s^{-1}}\).{ \footnote {Michaelis-Menten model is a classical multi-scale problem because the reversible reaction
 \(E+ S \xrightleftharpoons[d_{1}]{c_{1}} ES\)  is much faster than the reaction \(ES\stackrel{c_{2}}{\longrightarrow}  E+P\) by orders of magnitude. This situation will lead a static partitioning of the reactions.  Therefore, we choose the rate constants of the reaction rates such that we can simulate the model with dynamic partitioning of the reactions. } }
The initial number of molecules 
reveals in the conservation constants \(C=(50,300)^{T}\). 

The system is simulated over the time interval \(\left[0 , 100\right]\) seconds,  and the fixed stepsize for simulating the continuous SDE part is taken to be  $\Delta=0.3\mathrm{s}$. 
The  relative threshold error for partitioning the 
reactions is taken to be \(\varepsilon=0.25\), and the reaction set is repartitioned after every \(P=50\) iterations of updating the state vector.
As pointed out before, the SDE part of the approximating hybrid diffusion process is simulated by
a Runge-Kutta method of strong order \(2\) as outlined in Section \ref{sec:conservation}. Figure \ref{fig:MM_dynamic}  depicts the types of reactions and  a single realization of the model when Algorithm \ref{sumalgo} is applied.  
Figure \ref{fig:qq_MM_scaled_50} compares  the probability distributions and Q-Q plots of the states  $S$ and $P$ at time $t=60\mathrm{s}$ obtained from 
simulating the exact CTMC model with Gillespie's algorithm and the hybrid diffusion model with Algorithm \ref{sumalgo} of Section \ref{sec:algorithm}.
It should be observed that the probability distributions and  Q-Q plots obtained from the approximate hybrid diffusion algorithm are remarkably close to those obtained from the exact Gillespie algorithm demonstrating the accuracy of our dynamic partitioning algorithm. Evidently, the accuracy can be further increased by lowering the threshold \(\varepsilon\) of the error bound (\ref{error_bound1}). This would lead 
 to partitioning where most reactions for most of the time are treated as discrete reactions.



\begin{figure}[H]
\centering
\includegraphics[width=1\linewidth]{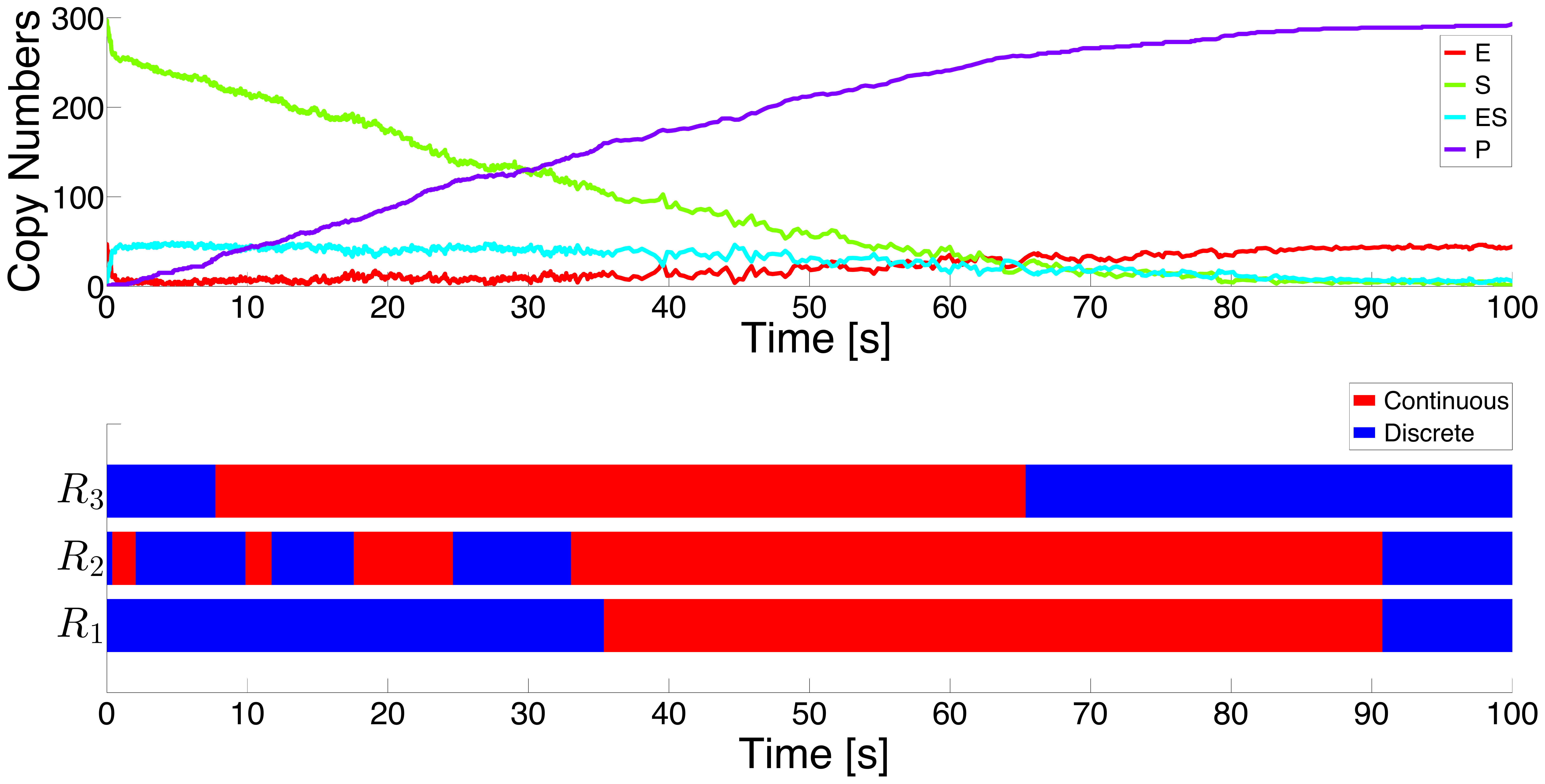}
\caption{Illustration of the hybrid diffusion algorithm with dynamic partitioning of reaction channels. 
 Upper Panel:  A single realization of the  Michaelis-Menten model when the fast reactions are modeled by diffusion approximations. Lower Panel: 
The portions of time when a reaction is treated as fast (continuous) or slow (discrete) as dictated by the error bound  (\ref{error_bound1}).\label{fig:MM_dynamic}}
\end{figure}

\begin{figure}[H]
\centering
\includegraphics[width=1\textwidth]{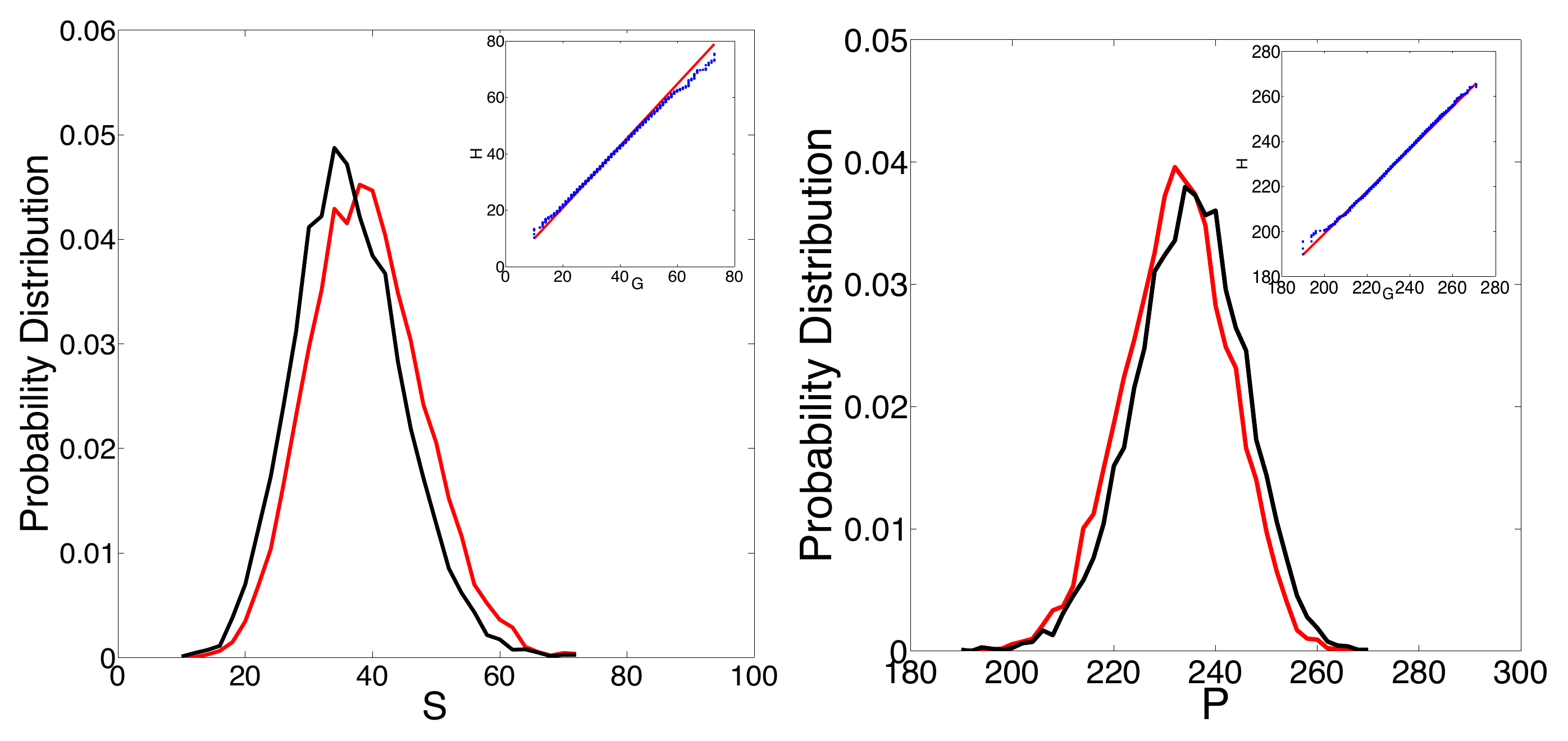}
\caption{Probability distributions of  \(S\) (left), \(P\) (right) at  time \(t=60\mathrm{s}\)  from \(30000\) samples constructed  with  (i) the Gillespie's algorithm (black line)  and by (ii) the  hybrid diffusion algorithm (red line). Insets show  Q-Q plot of   \(30000\) samples comparing the Gillespie's algorithm (G) and the hybrid diffusion algorithm (H).\label{fig:qq_MM_scaled_50} }
\end{figure}

\subsection{Validating the Error Bound}
\label{application:MM2}
  
\vs{.2cm}
Recall that the error bound obtained in Theorem \ref{thm:diffussion} has two parts of which $\Upsilon_k$, defined by \eqref{error_sde1}, captures the effect of treating reaction $R_k$ as a fast reaction and simulating it by a diffusion approximation. It should be noted that $\Upsilon_k$ depends on the initial condition. For an error bound to be meaningful for an appropriate scheme, it is desirable that it is sensitive, meaning that if the actual error increases or decreases, the bound behaves accordingly. Theoretically, in this case this means that for a fixed time interval $[0,t]$, we want $\Upsilon_k$ to be a non-decreasing function of $e_k(t) = |\EE[X^N(t)]-\EE[Z^N(t)]|$ when they both vary with respect to different initial conditions. Recall that $N$ is a scaling parameter which was determined according to a particular given initial state. The suffix $k$ is used to signify  that the reaction $R_k$ is simulated by a diffusion approximation. 
For a particular reaction system, this could be effectively checked  by plotting $\hat{e}_k$ (a Monte-Carlo estimate of $e_k$) versus $\hat{\Upsilon}_k$ (the  Monte-Carlo estimate of $\Upsilon_k$ given by \eqref{error_bound1})  for different initial conditions.

For the present Michaelis-Menten model, 
 \(R_1\) is considered as a fast (continuous) reaction while  \(R_2\) and  \(R_3\) are kept as slow (discrete)
ones (see Section \ref{ssec:sde}). The initial values of $E, ES$ and $P$ are kept fixed at \(E(0)=10\), \(ES(0)=30\) and \(P(0)=0\), while the initial values of the substrate $S$ are varied over $9$ different values from $\{25,30,35,40,45,50,55,60,65\}$. The time interval is taken to be $[0,1] \mathrm{s} $ and the fixed stepsize for simulating the continuous SDE part is taken as $\Delta=0.1\mathrm{s} $. The rate constants are given by \(c_{1}=0.02\) \(\mathrm{molec^{-1}s^{-1}}\), \(d_{1}=0.5\) \(\mathrm{s^{-1}}\),
\(c_{2}=0.1\) \(\mathrm{s^{-1}}\), for reactions in  (\ref{MM_reactions}). For $M$ realizations of the exact process $X^{N}$ and the approximating process $Z^{N}$, $\hat{e}_1$ was calculated by the usual Monte-Carlo average: 

\begin{equation}
\label{diff_ex_arnab}
\hat{e}_1=\sum_{i=1}^{4}  |\f{1}{M} \displaystyle \sum_{j=1}^{M} \displaystyle X_{ij}^{N}(t) - \f{1}{M} \displaystyle \sum_{j=1}^{M} \displaystyle  Z_{ij}^{N}(t)| \mbox{ with } t=1\mathrm{s},
\end{equation}
where \(X_{ij}^{N}\), \(Z_{ij}^{N}\) denote the number of molecules of the \(i\)-th species for the \(j\)-th realization for the corresponding processes. The result presented in Figure \ref{static_error_bound2} 
 demonstrates the desired monotone increasing property of the error bound.

\begin{figure}[H]
\centering
\includegraphics[width=0.8\textwidth]{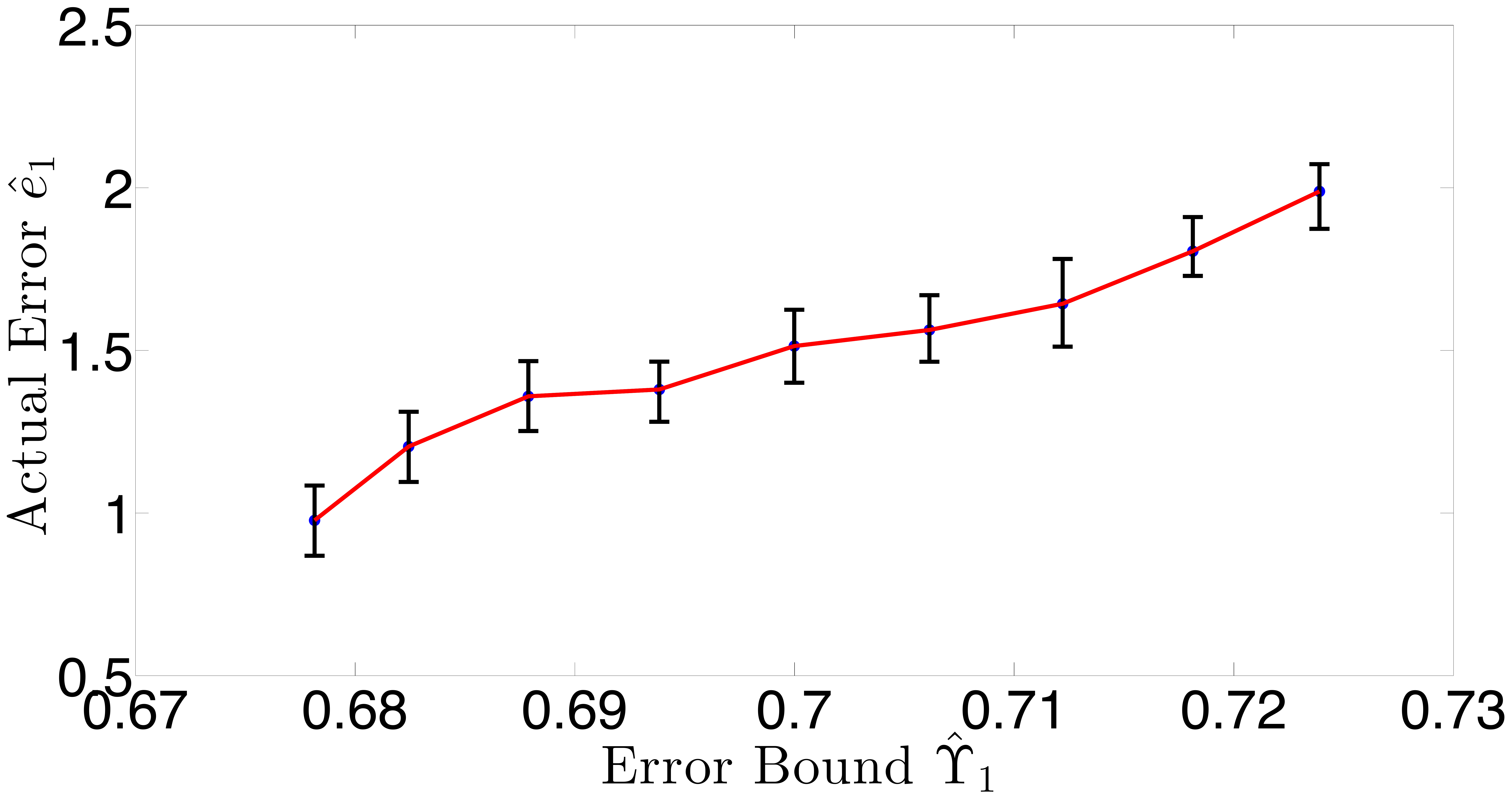}
\caption{ Monotone relation between actual error  $\hat{e}_1$  and the error bound (\ref{error_bound1})   for the Michaelis-Menten model (\ref{MM_reactions}). Monte-Carlo estimates are obtained by the Gillespie's direct  method  and by the  hybrid diffusion  algorithm static partitioning  for different initial conditions when \(R_{1}\) in  (\ref{MM_reactions}) is simulated by a diffusion approximation.
 We compute the error coming from \(R_1\)  through (\ref{error_bound1})  (\(x\) axis)   and the difference of expectations through 
(\ref{diff_ex_arnab}) (\(y\) axis) for each  initial condition
for  \(M=50000\) realizations  at  \(t=1\mathrm{s}\). The blue dots denote $\hat{e}_1$ for initial values and 
the error bars 
represent the confidence interval with lower error bound using a percentage of \(5 \%\) while upper  error bound using a percentage of \(95\%\).  \label{static_error_bound2} 
}
\end{figure}

\subsection{Lotka-Volterra Model }
\label{application:LV}
\vs{.2cm}

For our next example, we consider the popular Lotka-Volterra model, also known as the predator-prey system.
The model describes the dynamics of an  abstract environmental system where two animal species interact. 
Let \(S_1\) and \(S_2\) denote the prey and predator, respectively, the corresponding reaction system is given by 

\begin{equation}
\label{LV_reactions}
S_{1}\stackrel{c_{1}}{\longrightarrow} 2 S_{1},\quad  S_{1}+S_{2} \stackrel{c_{2}}{\longrightarrow} 2 S_{2}   \quad S_{2} \stackrel{c_{3}}{\longrightarrow} \emptyset.   
\end{equation}
The state of the system is defined  by $X(t) \in \Z_{\geq 0}^{2}$ such that \(X_{1}(t),\:X_{2}(t)\) represent the numbers of prey and predators at time \(t>0\), respectively. 
In our simulation,  
\[X(0)=(900, 800)^{T}, \:c_1=2 \: \mathrm{s^{-1}},\:c_2=0.002 \: \mathrm{molec^{-1}s^{-1}},\:c_3=2 \: \mathrm{s^{-1}}.\]
The fixed stepsize for simulating the SDE part is taken as $\Delta=0.5\mathrm{s}$. The  relative threshold error for partitioning  
reactions is taken as \(\varepsilon=0.03\), and the reaction set is again repartitioned after every \(P=50\) iterations of updating the state vector.

Figure \ref{fig:LV_dynamic} demonstrates how  the proposed algorithm switches back and forth between exact and hybrid diffusion approximation depending on the state of the system. As in the case of Michaelis-Menten kinetics, the accuracy of our hybrid diffusion algorithm is evident from Figure \ref{fig:qq_LV_new_scaled_15}   which compares the probability distributions  and the corresponding Q-Q plots  obtained from the exact Gillespie's algorithm and Algorithm \ref{sumalgo}.


\begin{figure}[H]
\centering
\includegraphics[width=1\textwidth]{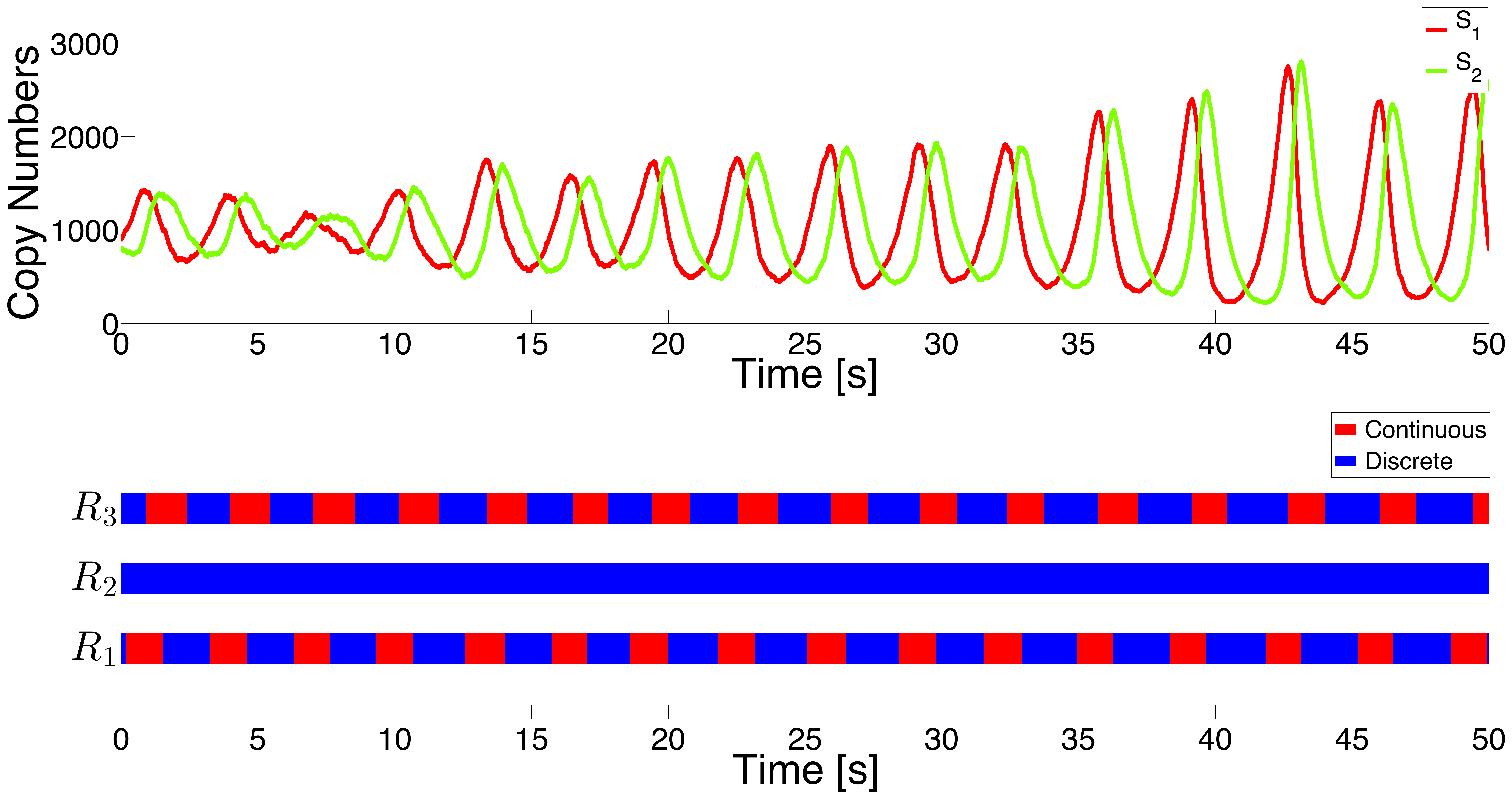}
\caption{Upper Panel:  A single realization of the  Lotka-Volterra model when the fast reactions are modeled by diffusion approximation. Lower Panel: Figure depicting the portions of time when a reaction is treated as fast (continuous) or slow (discrete) as dictated by the error bound  (\ref{error_bound1}).\label{fig:LV_dynamic}}
\end{figure}

\begin{figure}[H]
\centering
\includegraphics[width=1\textwidth]{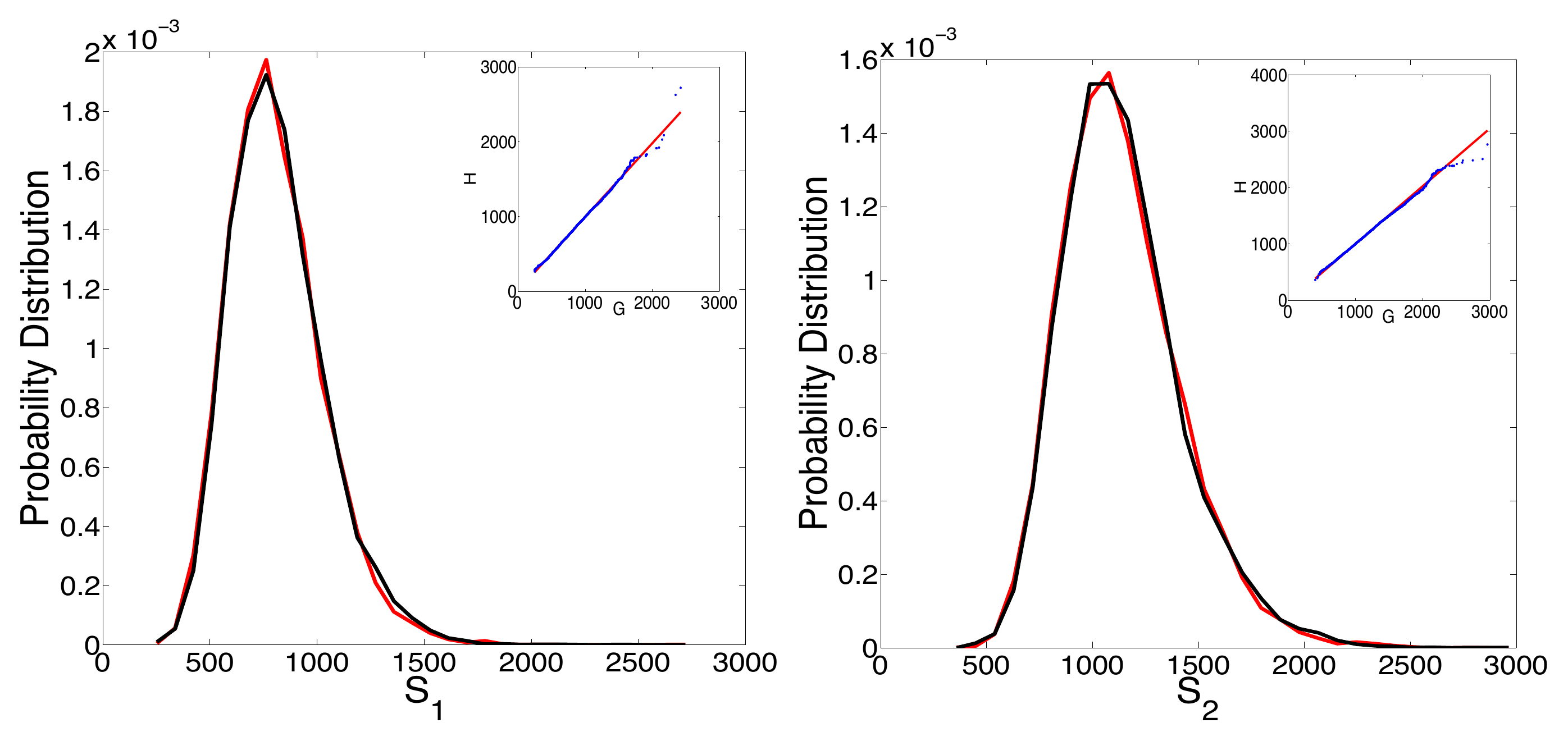}
\caption{Probability distributions of  \(S_1\) (left) and \(S_2\) (right)  at   \(t=15 \mathrm{s}\) from \(15000\) samples constructed  with  (i) the Gillespie's algorithm (black line)  and by (ii) the  hybrid diffusion algorithm (red line). Insets show  Q-Q plot of   \(15000\) samples comparing the Gillespie's algorithm (G) and the hybrid diffusion algorithm (H).\label{fig:qq_LV_new_scaled_15} }
\end{figure}

\subsection{The MAPK Pathway and Gene Expression Model }
\label{application:MAPK}
  
\vs{.2cm}
 The \(\mathrm{MAPK}\) (Mitogen Activated Protein Kinase) pathway is one of the 
most studied signal transduction mechanism that can be observed in all 
eukaryotic cells. \(\mathrm{MAPK}\) cascade conveys external signals from the cell membrane to 
the nucleus and regulates many cellular processes such as proliferation, differentiation, survival and motility.
The basic structure of a \(\mathrm{MAPK}\) cascade consists  of three kinases which are  the kinase kinase \(\mathrm{MAPKKK}\), the kinase \(\mathrm{MAPKK}\), 
and  the final kinase \(\mathrm{MAPK}\). The signaling process is initiated with a \(\mathrm{G-protein}\) which transmits the signal to \(\mathrm{MAPKKK}\). Phosphorylated \(\mathrm{MAPKKK}\) phosphorylates 
\(\mathrm{MAPKK}\) which activates \(\mathrm{MAPK}\). Kinase activation is defined by two reactions in analogy  to the Michaelis-Menten model of Section
\ref{application:MM1}. The first reaction is a reversible reaction which expresses the binding 
of a kinase to its substrate to form a complex, and the second reaction converts this complex to a kinase and an activated substrate. (see \cite{kol:00,kcg:05,osvcgk:05,rs:05,ts:02}
and the references therein) . 

Each \(\mathrm{MAPK}\) cascade is named according to their  \(\mathrm{MAPK}\)  components \cite{rs:05}. 
In this paper, we will study the \(\mathrm{ERK}\) (Extracellular Signal Regulated Kinase) pathway which 
includes  \(\mathrm{Ras}\) as  a \(\mathrm{G-protein}\), \(\mathrm{Raf}\) as \(\mathrm{MAPKKK}\), \(\mathrm{MEK}\)  as \(\mathrm{MAPKK}\) and \(\mathrm{ERK}\) as \(\mathrm{MAPK}\) \cite{kol:00}.
In our model, the process is initiated with \(\mathrm{\mathrm{Ras}-GTP}\) which is the active form of  \(\mathrm{Ras}\). It binds to \(\mathrm{Raf}\) to form 
\(\mathrm{Raf}:\mathrm{Ras-GTP} \) complex which in turn forms \(\mathrm{Raf}_{\mathrm{P}}\)  (phosphorylated \(\mathrm{Raf}\)). \(\mathrm{Raf}_{\mathrm{P}}\) binds \(\mathrm{MEK}\), \(\mathrm{MEK}_{\mathrm{P}}\)
to form  \(\mathrm{\mathrm{MEK}}_{\mathrm{P}}\), \(\mathrm{\mathrm{MEK}}_{\mathrm{PP}}\), respectively. Finally, \(\mathrm{\mathrm{MEK}}_{\mathrm{PP}}\) binds \(\mathrm{ERK}\) and \(\mathrm{ERK}_{\mathrm{P}}\) 
to form \(\mathrm{ERK}_{\mathrm{P}}\)  and \(\mathrm{ERK}_{\mathrm{PP}}\) respectively. Phosphatases, namely \(\mathrm{Pase1}\), \(\mathrm{Pase2}\) and \(\mathrm{Pase3}\) deactivate 
\(\mathrm{Raf}_{\mathrm{P}}\), deactivate \(\mathrm{MEK}\)'s 
 (i.e. \( \mathrm{MEK}_{\mathrm{P}} , \mathrm{MEK}_{\mathrm{PP}}\)) and deactivate  \(\mathrm{ERK}\)'s  (i.e. \(\mathrm{ERK}_{\mathrm{P}}\), \(\mathrm{ERK}_{\mathrm{PP}}\) ), respectively. 

The aim of the \(\mathrm{ERK}\) signaling pathway is to transform extracellular 
signals into intracellular signals and finally into a gene regulatory response. 
External stimulus 
activates a cell surface receptor which in turn  initiates the  
\(\mathrm{ERK}\) pathway in the cell. The transcriptional factor \(\mathrm{ERK}_{\mathrm{PP}}\) can then  change a target gene from an 
inactive form  \(\mathrm{GENE}_{\mathrm{off}}\)
to an active form  \(\mathrm{GENE}_{\mathrm{on}}\). 
For this reason, in our model, the 
reaction rate of the activation reaction, \(d_{21}(t)\), is defined as \(d_{21}(t)=d_{21}^{0}+d_{21}^{1} \left[\mathrm{ERK}_{\mathrm{PP}}\right(t)]\), where  \(d_{21}^{0}\),  \(d_{21}^{1}\) are  constants and \([\mathrm{ERK}_{\mathrm{PP}}(t)]\)
denotes the concentration of  \(\mathrm{ERK}_{\mathrm{PP}}\)  at time $t$. 
The active gene leads to the production of \(\mathrm{mRNA}\) (transcription) and \(\mathrm{mRNA}\) is further processed
into a \(\mathrm{Protein}\)  through the process of translation. The increase in the number of target \(\mathrm{Protein}\)  corresponds to the cellular response to the initial 
extracellular signal 
(see Figure \ref{fig:mapk_graphical}A) \cite{klwklh:09} .

\begin{figure}[H]
\centering
\includegraphics[width=1\textwidth]{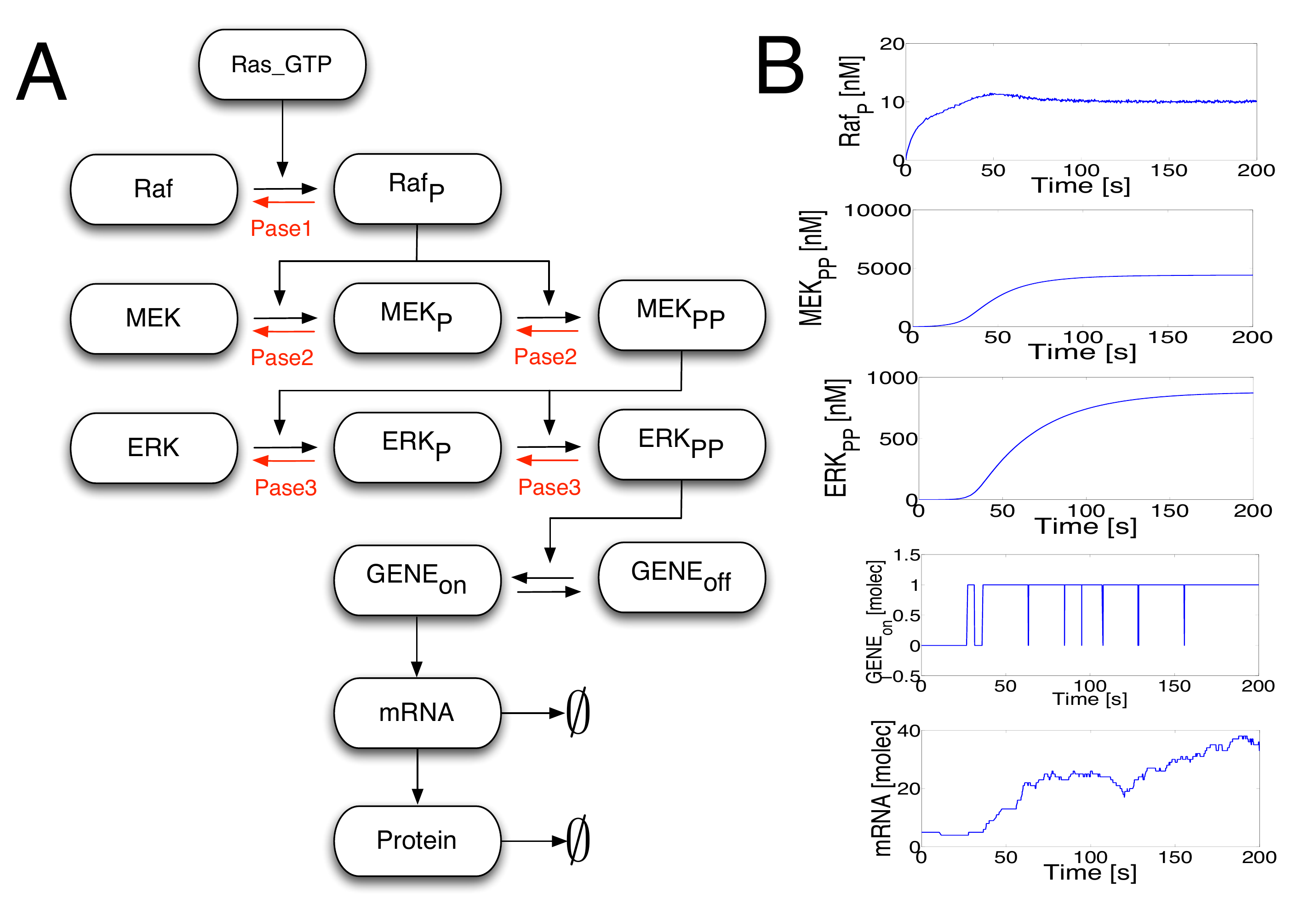}
\caption{ (A) Systematic representation of  \(\mathrm{MAPK}\) pathway and its coupled gene expression.\label{fig:mapk_graphical} 
(B) Concentrations of the species \(\mathrm{Raf}_{\mathrm{P}},\:\mathrm{MEK}_{\mathrm{PP}},\:\mathrm{ERK}_{PP}\)
and the copy numbers of  \(\mathrm{GENE}_{\mathrm{on}}\), \(\mathrm{mRNA}\) in \(t \in \left[0,200 \right] \mathrm{s}\) when   \(R_1,\ldots,R_{20}\) reactions in Table \ref{table_mapk1} are considered as fast reactions and  simulated by diffusion approximation while \(R_{21},\ldots,R_{25}\)  are considered as slow reactions and Markov chain formulation is kept. }
\end{figure}
This is a typical realistic example in modeling coupled cellular 
processes in systems of biology. Due to the combination of fast, high-copy signal transduction with slow, low copy gene expression dynamics
this is a typical multi-scale, stiff problem for which our hybrid algorithm is designed for.
The complete list of reactions and the reaction rate  constants in  \(\mathrm{\left[s^{-1}\right]}\) and  \(\mathrm{\left[molec^{-1}s^{-1}\right]}\)
for unimolecular and bimolecular  reactions, respectively,  for that model are given in  Table \(\ref{table_mapk1}\). The reactions and
their rates are taken from  \cite{csjnnls:09}. 
\newpage
\renewcommand{\arraystretch}{1.5}
\begin{table}[H]

\scriptsize
\begin{center}
\begin{tabular}{|c | c | }
\hline Reactions &Stochastic rate constants \\\hline    

\(R_1:\:\mathrm{Ras-GTP}+ \mathrm{Raf} \xrightleftharpoons[d_{1}]{c_{1}} \mathrm{Raf}:\mathrm{Ras-GTP} \)&\(c_{1}=5e-06 \: \mathrm{molec^{-1}s^{-1}} \)\\
&\(d_{1}=0.0053 \: \mathrm{s^{-1}} \)\\\hline 

\(R_2:\:\mathrm{Raf}:\mathrm{Ras-GTP}\stackrel{c_{2}}{\longrightarrow} \mathrm{Ras-GTP}+\mathrm{Raf}_{\mathrm{P}}\)&\(c_{2}=3.1  \: \mathrm{s^{-1}}  \)\\\hline

\(R_3:\:\mathrm{Raf}_{\mathrm{P}}+\mathrm{Pase1} \xrightleftharpoons[d_{3}]{c_{3}} \mathrm{Raf}_{\mathrm{P}}:\mathrm{Pase1} \)&\(c_{3}=6e-05\: \mathrm{molec^{-1}s^{-1}}  \)\\ 
&\(d_{3}=1.41589e-04\: \mathrm{s^{-1}} \)\\\hline 

\(R_4:\: \mathrm{Raf}_{\mathrm{P}}:\mathrm{Pase1}  \stackrel{c_{4}}{\longrightarrow} \mathrm{Raf}+\mathrm{Pase1}\) &\(c_{4}=3.16228\: \mathrm{s^{-1}} \)\\\hline 

\(R_5:\: \mathrm{Raf}_{\mathrm{P}}+\mathrm{MEK}\xrightleftharpoons[d_{5}]{c_{5}} \mathrm{Raf}_{\mathrm{P}}:\mathrm{MEK} \)&\(c_{5}=1.07e-05 \: \mathrm{molec^{-1}s^{-1}} \)\\ 
&\(d_{5}=0.033 \: \mathrm{s^{-1}} \)\\\hline

\(R_6:\:\mathrm{Raf}_{\mathrm{P}}:\mathrm{MEK}\stackrel{c_{6}}{\longrightarrow} \mathrm{Raf}_{\mathrm{P}}+\mathrm{MEK}_{\mathrm{P}}\)&\(c_{6}=1.9 \: \mathrm{s^{-1}} \)\\\hline

\(R_7:\: \mathrm{MEK}_{\mathrm{P}}+\mathrm{Pase2}\xrightleftharpoons[d_{7}]{c_{7}} \mathrm{MEK}_{\mathrm{P}}:\mathrm{Pase2}\)&\(c_{7}=4.74801e-08 \: \mathrm{molec^{-1}s^{-1}} \)\\ 
&\(d_{7}=0.252982 \: \mathrm{s^{-1}} \)\\\hline 

\(R_8:\: \mathrm{MEK}_{\mathrm{P}}:\mathrm{Pase2} \stackrel{c_{8}}{\longrightarrow} \mathrm{MEK}+\mathrm{Pase2}\)&\(c_{8}=0.112387 \: \mathrm{s^{-1}} \)\\\hline 

\(R_9:\: \mathrm{Raf}_{\mathrm{P}}+ \mathrm{MEK}_{\mathrm{P}} \xrightleftharpoons[d_{9}]{c_{9}} \mathrm{Raf}_{\mathrm{P}}:\mathrm{MEK}_{\mathrm{P}} \)&\(c_{9}=1.07e-05\: \mathrm{molec^{-1}s^{-1}} \)\\ 
&\(d_{9}=0.033\: \mathrm{s^{-1}} \)\\\hline

\(R_{10}:\: \mathrm{Raf}_{\mathrm{P}}:\mathrm{MEK}_{\mathrm{P}} \stackrel{c_{10}}{\longrightarrow} \mathrm{Raf}_{\mathrm{P}}+\mathrm{MEK}_{\mathrm{PP}}\)&\(c_{10}=0.8 \: \mathrm{s^{-1}} \)\\\hline 

\(R_{11}:\: \mathrm{MEK}_{\mathrm{PP}}+\mathrm{ERK}\xrightleftharpoons[d_{11}]{c_{11}} \mathrm{MEK}_{\mathrm{PP}}:\mathrm{ERK}\)&\(c_{11}=8.85125e-07 \: \mathrm{molec^{-1}s^{-1}} \)\\ 
&\(d_{11}=0.01833 \: \mathrm{s^{-1}} \)\\\hline

\(R_{12}:\: \mathrm{MEK}_{\mathrm{PP}}:\mathrm{ERK}\stackrel{c_{12}}{\longrightarrow} \mathrm{MEK}_{\mathrm{PP}}+ \mathrm{ERK}_{\mathrm{P}}\)&\(c_{12}=0.028\: \mathrm{s^{-1}} \)\\\hline

\(R_{13}:\: \mathrm{ERK}_{\mathrm{PP}}+\mathrm{\mathrm{Pase3}}\xrightleftharpoons[d_{13}]{c_{13}} \mathrm{ERK}_{\mathrm{PP}}:\mathrm{\mathrm{Pase3}}\)&\(c_{13}=3.9739e-04 \: \mathrm{molec^{-1}s^{-1}} \)\\ 
&\(d_{13}=5 \: \mathrm{s^{-1}} \)\\\hline

\(R_{14}:\:\mathrm{ERK}_{\mathrm{PP}}:\mathrm{\mathrm{Pase3}} \stackrel{c_{14}}{\longrightarrow} \mathrm{ERK}_{\mathrm{P}}+\mathrm{\mathrm{Pase3}}\)&\(c_{14}=0.0076 \: \mathrm{s^{-1}} \)\\\hline

\(R_{15}:\:\mathrm{MEK}_{\mathrm{PP}}+\mathrm{Pase2}\xrightleftharpoons[d_{15}]{c_{15}} \mathrm{MEK}_{\mathrm{PP}}:\mathrm{Pase2}\)&\(c_{15}=2.37e-05 \: \mathrm{molec^{-1}s^{-1}} \)\\ 
&\(d_{15}=0.79 \: \mathrm{s^{-1}} \)\\\hline

\(R_{16}:\:\mathrm{MEK}_{\mathrm{PP}}:\mathrm{Pase2}  \stackrel{c_{16}}{\longrightarrow} \mathrm{\mathrm{MEK}}_{\mathrm{P}}+\mathrm{Pase2} \)&\(c_{16}=0.112387 \: \mathrm{s^{-1}} \)\\\hline

\(R_{17}:\: \mathrm{MEK}_{\mathrm{PP}}+\mathrm{ERK}_{\mathrm{P}}  \xrightleftharpoons[d_{17}]{c_{17}} \mathrm{MEK}_{\mathrm{PP}}:\mathrm{ERK}_{\mathrm{P}}\)&\(c_{17}=8.85125e-06 \: \mathrm{molec^{-1}s^{-1}} \)\\ 
&\(d_{17}=0.01833 \: \mathrm{s^{-1}} \)\\\hline

\(R_{18}:\:\mathrm{MEK}_{\mathrm{PP}}:\mathrm{ERK}_{\mathrm{P}}\stackrel{c_{18}}{\longrightarrow} \mathrm{MEK}_{\mathrm{PP}}+ \mathrm{ERK}_{\mathrm{PP}}\)&\(c_{18}=0.701662\: \mathrm{s^{-1}} \)\\\hline

\(R_{19}:\:\mathrm{ERK}_{\mathrm{P}}+\mathrm{\mathrm{Pase3}}\xrightleftharpoons[d_{19}]{c_{19}} \mathrm{ERK}_{\mathrm{P}}:\mathrm{\mathrm{Pase3}}\)&\(c_{19}=8.33e-07\: \mathrm{molec^{-1}s^{-1}} \)\\ 
&\(d_{19}=0.25 \: \mathrm{s^{-1}} \)\\\hline

\(R_{20}:\:\mathrm{ERK}_{\mathrm{P}}:\mathrm{\mathrm{Pase3}} \stackrel{c_{20}}{\longrightarrow} \mathrm{ERK}+\mathrm{\mathrm{Pase3}}\)&\(c_{20}=0.0076\: \mathrm{s^{-1}} \)\\\hline

\(R_{21}:\:\mathrm{GENE}_\mathrm{on}\xrightleftharpoons[d_{21}]{c_{21}} \mathrm{GENE}_\mathrm{off}\)&\(c_{21}=0.05 \: \mathrm{s^{-1}} \)\\ 
&\(d_{21}(t)=0.01+0.003 \left[\mathrm{ERK}_{\mathrm{PP}}(t)\right] \: \mathrm{s^{-1}} \)\\\hline

\(R_{22}:\:\mathrm{GENE}_\mathrm{on}\stackrel{c_{22}}{\longrightarrow} \mathrm{GENE}_\mathrm{on}+\mathrm{mRNA}\)&\(c_{22}=0.5 \: \mathrm{s^{-1}} \)\\\hline

\(R_{23}:\:\mathrm{mRNA}\stackrel{c_{23}}{\longrightarrow} \mathrm{mRNA}+\mathrm{Protein}\)&\(c_{23}=0.3\: \mathrm{s^{-1}} \)\\\hline 

\(R_{24}:\:\mathrm{mRNA}\stackrel{c_{24}}{\longrightarrow} \emptyset \)&\(c_{24}=0.015\: \mathrm{s^{-1}} \)\\\hline

\(R_{25}:\:\mathrm{Protein} \stackrel{c_{25}}{\longrightarrow} \emptyset \)&\(c_{25}=5e-06 \: \mathrm{s^{-1}} \)\\ \hline
\end{tabular}  
\end{center}   
\caption{Reactions and rate parameters for the  \(\mathrm{ERK}\) signal  gene expression pathway. \label{table_mapk1}}
\end{table} 
In our application, we will fix  the concentrations of  reactants of the MAPK cascade and the copy numbers of  reactants of the gene expression. Therefore, for the sake of simplicity,  we define the state vector as 
\(X(t)=(X_{\mathrm{MAPK}}(t),X_{\mathrm{GENE}}(t))^{T}\)
where 
\begin{eqnarray*}
X_{\mathrm{MAPK}}(t)&=&(\mathrm{Ras-GTP},\mathrm{Raf}, \mathrm{Raf}:\mathrm{Ras-GTP}, \mathrm{\mathrm{Raf}}_{\mathrm{P}},\mathrm{Pase1},\\
&&\mathrm{\mathrm{Raf}}_{\mathrm{P}}:\mathrm{Pase1},\mathrm{MEK},\mathrm{\mathrm{Raf}}_{\mathrm{P}}:\mathrm{MEK},\mathrm{\mathrm{MEK}}_{\mathrm{P}},\\
&&\mathrm{Pase2},\mathrm{\mathrm{MEK}}_{\mathrm{P}}:\mathrm{Pase2},\mathrm{\mathrm{Raf}}_{\mathrm{P}}:\mathrm{\mathrm{MEK}}_{\mathrm{P}},\mathrm{MEK}_{PP},ERK,\\
&&\mathrm{MEK}_{PP}:ERK,\mathrm{ERK}_{\mathrm{P}},\mathrm{ERK}_{\mathrm{PP}},\mathrm{MEK}_{\mathrm{PP}}:\mathrm{ERK}_{\mathrm{P}},\\
&&\mathrm{\mathrm{Pase3}},\mathrm{ERK}_{\mathrm{P}}:\mathrm{\mathrm{Pase3}},\mathrm{MEK}_{\mathrm{PP}}:\mathrm{Pase2},\mathrm{ERK}_{\mathrm{PP}}:\mathrm{\mathrm{Pase3}}),
\end{eqnarray*}
\(X_{\mathrm{GENE}}(t)=(\mathrm{GENE}_\mathrm{on},\mathrm{GENE}_\mathrm{off},\mathrm{mRNA},\mathrm{Protein})\). 
To convert the amount of species of the components \(X(t)\) from copy numbers to nanomolar 
\(\mathrm{\left[nM \right]}\) concentrations, 
 we use the relation 
\begin{equation}
\label{conversion}
U(t)=X(t)/n_{A} \Omega, 
\end{equation}
where \(n_{A}=6\times 10^{23}\) represents the 
Avogadro's number  and \(\Omega \) is the volume of the reaction compartment  in liters \cite{spfhxnlkpbgkwlfkn:09}.
A complete list of initial number of molecules and corresponding 
initial concentrations can be seen in Table \ref{table_mapk2}.

\begin{table}[H]
\footnotesize 
\begin{center}
\begin{tabular}{|c | c | c |}
\hline Species &Copy numbers& Concentrations [nM] \\\hline    
\(\mathrm{Ras-GTP} \)&10000 &16.6667\\\hline
\(\mathrm{Raf}\)&711081&1185.135\\\hline
\(\mathrm{Raf}:\mathrm{Ras-GTP}\)& 11&0.0183\\\hline
\(\mathrm{\mathrm{Raf}}_{\mathrm{P}}\)&11&0.0183\\\hline
\(\mathrm{Pase1}\)&49990&83.3167\\\hline
\(\mathrm{Raf}_{\mathrm{P}}:\mathrm{Pase1}\)&10&0.0167\\\hline
\(\mathrm{MEK}\)&2830668&4717.78\\\hline
\(\mathrm{\mathrm{Raf}}_{\mathrm{P}}:\mathrm{MEK}\)&173& 0.2883\\\hline
\(\mathrm{\mathrm{MEK}}_{\mathrm{P}}\)&185198&308.6633\\\hline
\(\mathrm{Pase2}\)&121372&202.2867\\\hline
\(\mathrm{\mathrm{MEK}}_{\mathrm{P}}:\mathrm{Pase2}\)&2921& 4.8683\\\hline
\(\mathrm{\mathrm{Raf}}_{\mathrm{P}}:\mathrm{\mathrm{MEK}}_{\mathrm{P}}\)&26&0.0433\\\hline
\(\mathrm{MEK}_{\mathrm{PP}}\)&59&0.0983\\\hline
\(\mathrm{ERK}\)&685746& 1142.91\\\hline
\(\mathrm{MEK}_{\mathrm{PP}}:\mathrm{ERK}\)&768&1.28\\\hline
\(\mathrm{ERK}_{\mathrm{P}}\)&5275&8.7917 \\\hline
\(\mathrm{ERK}_{\mathrm{PP}}\)& 27&0.045\\\hline
\(\mathrm{MEK}_{\mathrm{PP}}:\mathrm{ERK}_{\mathrm{P}}\)&0&0\\\hline
\(\mathrm{\mathrm{Pase3}}\)&165519& 275.865 \\\hline
\(\mathrm{ERK}_{\mathrm{P}}:\mathrm{\mathrm{Pase3}}\)&2823& 4.705\\\hline
\(\mathrm{MEK}_{\mathrm{PP}}:\mathrm{Pase2}\)&187& 0.3167 \\\hline
\(\mathrm{ERK}_{\mathrm{PP}}:\mathrm{\mathrm{Pase3}}\)&360&  0.6\\\hline
\(\mathrm{GENE}_\mathrm{on}\)&0& 0 \\\hline
\(\mathrm{GENE}_\mathrm{off}\)&1&0.0017\\\hline
\(\mathrm{mRNA}\)&5&0.0083\\\hline
\(\mathrm{Protein}\)& 0&0\\\hline
\end{tabular}  
\end{center}   
\caption{Species initial copy numbers and corresponding concentrations  for the nominal volume of the reaction compartment of \(\Omega=10^{-12}\) in liters \label{table_mapk2}.}
\end{table} 
According to the reactions in Table \ref{table_mapk1}, the following \(8\) conservation laws can be identified: 
\begin{eqnarray*}
C_{1}&=&\mathrm{Ras-GTP} +\mathrm{Raf}:\mathrm{Ras-GTP}\\
C_{2}&=&\mathrm{Raf}+\mathrm{\mathrm{Raf}}_{\mathrm{P}}+\mathrm{\mathrm{Raf}}_{\mathrm{P}}:\mathrm{MEK}+\mathrm{\mathrm{Raf}}_{\mathrm{P}}:\mathrm{Pase1}+\mathrm{\mathrm{Raf}}_{\mathrm{P}}:\mathrm{\mathrm{MEK}}_{\mathrm{P}}\\
     &+&\mathrm{Raf}:\mathrm{Ras-GTP}\\
C_{3}&=&\mathrm{Pase1}+\mathrm{\mathrm{Raf}}_{\mathrm{P}}:\mathrm{Pase1}\\
C_{4}&=&\mathrm{MEK}+\mathrm{\mathrm{MEK}}_{\mathrm{P}}+\mathrm{MEK}_{\mathrm{PP}}+\mathrm{\mathrm{MEK}}_{\mathrm{P}}:\mathrm{Pase2}\\
&+&\mathrm{MEK}_{\mathrm{PP}}:\mathrm{ERK}+\mathrm{MEK}_{\mathrm{PP}}:\mathrm{Pase2}+\mathrm{MEK}_{\mathrm{PP}}:\mathrm{ERK}_{\mathrm{P}}\\
&+&\mathrm{\mathrm{Raf}}_{\mathrm{P}}:\mathrm{MEK}+\mathrm{\mathrm{Raf}}_{\mathrm{P}}:\mathrm{\mathrm{MEK}}_{\mathrm{P}}\\
C_{5}&=&\mathrm{\mathrm{MEK}}_{\mathrm{P}}:\mathrm{Pase2}+\mathrm{\mathrm{MEK}}_{\mathrm{PP}}:\mathrm{Pase2}+\mathrm{Pase2}\\
C_{6}&=& \mathrm{ERK} + \mathrm{ERK}_{\mathrm{P}}+ \mathrm{ERK}_{\mathrm{PP}} + \mathrm{ERK}_{\mathrm{P}}:\mathrm{\mathrm{Pase3}} \\
&+& \mathrm{ERK}_{\mathrm{PP}}:\mathrm{\mathrm{Pase3}} + \mathrm{MEK}_{\mathrm{PP}}:\mathrm{ERK} + \mathrm{MEK}_{\mathrm{PP}}:\mathrm{ERK}_{\mathrm{P}}\\
C_{7}&=&\mathrm{ERK}_{\mathrm{P}}:\mathrm{\mathrm{Pase3}}+\mathrm{ERK}_{\mathrm{PP}}:\mathrm{Pase3}+\mathrm{Pase3}\\
C_{8}&=&\mathrm{GENE}_\mathrm{on}+\mathrm{GENE}_\mathrm{off}.\\
\end{eqnarray*}
Here, $C_{1},C_{2},\ldots,C_{8}$ are mass conservation constants.
Since the number of molecules of \(\mathrm{MAPK}\) signaling cascade, \(X_{\mathrm{MAPK}}(t)\) 
are  high, we expect  \(R_1,\ldots,R_{20}\) reactions in Table \ref{table_mapk1}  will be 
 fast reactions while  \(R_{21},\ldots,R_{25}\)  will be  slow reactions. 
Concentrations of \(\mathrm{Raf}_\mathrm{P}\), \(\mathrm{MEK}_{\mathrm{PP}}\), \(\mathrm{ERK}_{\mathrm{PP}}\) and the copy numbers of  
\(\mathrm{GENE}_\mathrm{on}\), \(\mathrm{mRNA}\)  when \(R_1,\ldots,R_{20}\) reactions in Table \ref{table_mapk1} are considered as continuous reactions 
and modeled by diffusion approximation while \(R_{21},\ldots,R_{25}\)  are considered as slow reactions and 
modeled by Markov jump process  in time interval \(\left[0,200 \right] \mathrm{s}\)
can be seen in Figure \ref{fig:mapk_graphical}B. 
This figure demonstrates the amounts of some species when static partitioning algorithm is applied. 
We also implement dynamic partitioning algorithm to the model and compare the computation time of  Algorithm \ref{sumalgo} with  that of 
Gillespie's direct method for different volumes of the reaction compartment  \(\Omega\). In our application, we consider the volume of the cell compartment 
where the signaling reactions take place, is scaled with a positive constant \(\omega\) such that \(\Omega=\omega V\)
where \(V=10^{-12}\) is the nominal 
cell volume in liters \cite{spfhxnlkpbgkwlfkn:09} while the sub-compartment where the gene expression reactions take place
is fixed for all \(\Omega\) values. We consider  concentrations of the reactants 
of the \(\mathrm {MAPK}\) cascade, \(X_{\mathrm{MAPK}}(t)\), and the number of molecules of species of the gene expression, \(X_{\mathrm{GENE}}(t)\),  are same 
for all volumes (see Table \ref{table_mapk2}). Hence, to keep same concentrations for all volumes, we multiply the number of molecules of the components \(X_{\mathrm{MAPK}}(t)\) with \(\omega\) 
if the volume of the reaction compartment is  \(\Omega=\omega V\). Also, it must 
be noticed that the stochastic rate constants of bimolecular reactions are divided by \(\omega\)
when the volume of reaction compartment  is \(\Omega=\omega V\). 
In Table \ref{table:time}, one can see the CPU times of our algorithm  (see Algorithm \ref{sumalgo}) 
and Gillespie's direct method  in seconds for a single  realization of the model in time interval \([0,10]\mathrm{s}\) with \(\Delta=2\mathrm{s} ,\: \varepsilon=0.3,\:P=20\)
for \(\omega=1,2,5,10,10^{2},10^{4},10^{6},10^{8}\). 
To obtain the numerical solution of SDEs, we automatically  choose the stepsize of Runge-Kutta method as explained in 
Section \ref{sec:runge_kutta} with absolute and relative  tolerances
of \(10^{-12}\) and \(10^{-9}\), respectively. Since both the drift, the  diffusion term of SDE  given by (\ref{eq:11})  are 
volume dependent, the CPU time  of  the hybrid  diffusion algorithm decreases as expected when the volume increases. 
Types of reactions according to Algorithm \ref{sumalgo} for \(\omega=1\) can be seen in Figure \ref{fig:reaction_types}.
As discussed before, to preserve the same concentration of species for all volumes, we have to multiply the number of molecules of some species with \(\omega\)
to observe the dynamics for volume  \(\Omega=\omega V\).  
This will make a significant change on the CPU time of the Gillespie's  algorithm. Since application of  the Gillespie's direct method is 
too time consuming for large volumes, we compute its CPU time only for \(\omega=1,2,5,10\). The results 
reveal that  the hybrid diffusion algorithm significantly reduce the computational time.

\renewcommand{\arraystretch}{1.75}
\begin{table}[H]
\footnotesize 
\begin{center}
\begin{tabular}{|c | c | c |c |}
\hline  \multirow{2}{*} {\(\omega\)}&\multicolumn{2}{p{10cm} |}{\centering\mbox{ CPU Time in seconds }} \\  \cline{2-3}
 & \mbox{ Algorithm \ref{sumalgo} } & \mbox{ Gillespie's Direct Method }\\\hline   
\(1 \)& \(498.82\)& \(2872.02\)\\\hline
\(2\) & \( 498.080 \)&\(6264.26\) \\\hline
\(5 \)&\( 404.36 \)& \(15613.46\)\\\hline
\(10\) & \(345.41 \)& \(31356.38\)\\\hline
\(10^{2}\) & \( 337.01 \)&\rule{.6cm}{0.4pt}\\\hline
\(10^{4} \)& \(167.61 \)& \rule{.6cm}{0.4pt}\\\hline
\(10^{6}\)& \( 119.93\)& \rule{.6cm}{0.4pt}\\\hline
\(10^{8}\)& \( 78.21\)& \rule{.6cm}{0.4pt}\\\hline
\end{tabular}  
\end{center}   
\caption{ The CPU times (in seconds) of  the Algorithm \ref{sumalgo} and the Gillespie's direct method for the \(\mathrm{MAPK}\) cascade together 
with its gene expression. For all volumes, we have the same concentrations  of the \(\mathrm{MAPK}\)  cascade species and the same copy numbers 
 for  species of  gene expression. To preserve the same concentration of \(\mathrm{MAPK}\)  species, we multiply their number of molecules with \(\omega\) to observe the dynamics 
of the model in the reaction compartment with volume \(\Omega=\omega V\) where \(\Delta=2 \mathrm{s}\), \(\varepsilon=0.3,\:P=20\) in \(t \in [0,10] \mathrm{s}\). \label{table:time}}
\end{table}

\begin{figure}[H]
\centering
\includegraphics[width=1\textwidth]{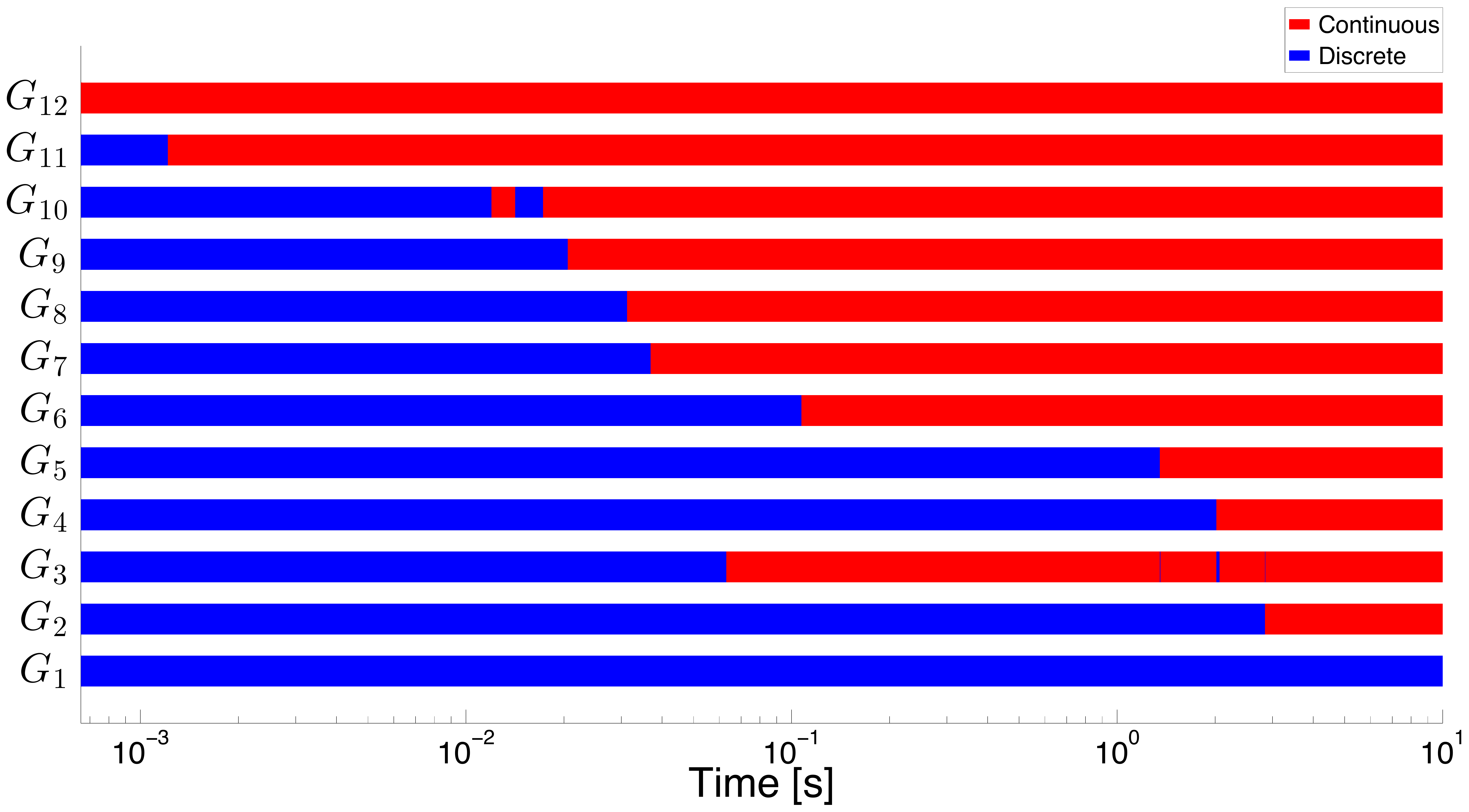}
\caption{ Depicts the portions of time when reactions of  the  \(\mathrm{MAPK}\)  signaling pathway are treated as fast or slow in the time interval \([0,10] \mathrm{s}\) according to our dynamic partitioning algorithm with \(\Delta=2\mathrm{s}\), \(\varepsilon=0.3\), \(P=20\) and \(\omega=1\). Since the number of reactions is too high, to have a better visualization we group the reactions 
and also note that x-axis represents the \(log(t)\) of time \(t>0\). Here, 
\(G_{1}=(R_{21}^{+},R_{21}^{-},R_{22},R_{23},R_{24},R_{25})\), \(G_{2}=R_{18}\), \(G_{3}=(R_{13}^{+},R_{13}^{-})\), \(G_{4}=(R_{17}^{+},R_{17}^{-})\), \(G_{5}=R_{3}^{-}\), \(G_{6}=R_{3}^{+}\),
\(G_{7}=R_{9}^{+}\), \(G_{8}=(R_{4},R_{10})\), \(G_{9}=(R_{2},R_{5}^{+})\), \(G_{10}=R_{6}\), \(G_{11}=(R_{1}^{+},R_{1}^{-})\), \(G_{12}=(R_{5}^{-},R_{7}^{+},R_{8},R_{9}^{-},R_{11}^{+},R_{11}^{-},R_{12},R_{14},R_{15}^{+},R_{15}^{-},R_{16},R_{19}^{+},R_{19}^{-},R_{20})\) where
\(R_{i}^{+}\) denotes the forward reaction \((\rightharpoonup)\) while   \(R_{i}^{-}\)  denotes the backward reaction \((\leftharpoondown)\) for \(i=1,2,\ldots,25\). \label{fig:reaction_types} }
\end{figure}

\section{Conclusion}
In this work, we developed techniques for simulating certain multi-scale reaction systems more efficiently. Our strategy involved a systematic approach 
of separating the reactions into fast and slow groups. The fast reactions are simulated by a diffusion approximation while the exact Gillespie-type simulation procedure 
is maintained for the slow ones. The partitioning of the reaction set is based on an appropriate error bound whose derivation is one of the central themes of the paper. The theoretical results are then effectively encoded in an efficient fast algorithm which also allows for the partition to change dynamically over the course of time. 

The paper used Runge-Kutta integration 
methods to approximate the solutions of SDEs.  Also the conservation relations, occurring naturally in many reaction models, have been properly utilized to reduce the dimensionality of the corresponding SDEs and to ensure strict mass conservation.

The proposed algorithms are implemented for the well-known Michaelis-Menten kinetics, the Lotka-Volterra system 
and a realistic MAPK cascade together with its gene expression. The results reveal that the proposed 
algorithm simulates the multi-scale processes with a significant gain in runtime but little loss in accuracy when compared to
the exact Gillespie algorithm.

\section{Acknowledgement}
\label{acknowledgement}
\textit{Funding} D. Alt\i ntan acknowledges the support from  the Scientific and Technological Research Council of Turkey (T{\"{U}}B{\.{I}}TAK), grant no. 2219.

\bibliographystyle{plain}
\bibliography{hybrid_bibliography}

\begin{thebibliography}{10}

\bibitem{actvh:005}
A.~Alfonsi, E.~Canc\`{e}s, G.~Turinici, B.D. Ventura, and W.~Huisinga.
\newblock Adaptive simulation of hybrid stochastic and deterministic models for
  biochemical systems.
\newblock In {\em ESAIM: Proc.}, volume~14, pages 1--13, 2005.

\bibitem{ak:11}
D.F. Anderson and T.G. Kurtz.
\newblock Continuous time markov chain models for chemical reaction networks.
\newblock In H.~Koeppl, G.~Setti, M.~di Bernardo, and D.~Densmore, editors,
  {\em Design and Analysis of Biomolecular Circuits}. Springer-Verlag, 2011.

\bibitem{bb:96}
K.~Burrage and P.M. Burrage.
\newblock High strong order explicit {R}unge-{K}utta methods for stochastic
  ordinary differential equations.
\newblock {\em Applied Numerical Mathematics}, 22:81--101, 1996.

\bibitem{csjnnls:09}
W.W. Chen, B.~Schoeberl, P.J. Jasper, M.~Niepel, U.B. Nielsen, and D.A.
  Lauffenburger.
\newblock Input--output behavior of {ErbB} signaling pathways as revealed by a
  mass action model trained against dynamic data.
\newblock {\em Molecular Systems Biology}, 5(239), 2009.

\bibitem{ch:02}
A.~Cornish-Bowden and J.-H.S. Hofmeyr.
\newblock The role of stoichiometric analysis in studies of metabolism : An
  example.
\newblock {\em Journal of Theoretical Biology}, 216:179--191, 2002.

\bibitem{ccsl:13}
A.~Coulon, C.~C. Chow, R.~H. Singer, and D.~R. Larson.
\newblock Eukaryotic transcriptional dynamics: from single molecules to cell
  populations molecules to cell populations.
\newblock {\em Nature Reviews Genetics}, 14:572--584, 2013.

\bibitem{cdr:09}
A.~Crudu, A.~Debussche, and O.~Radulescu.
\newblock Hybrid stochastic simplifications for multiscale gene networks.
\newblock {\em BMC Systems Biology}, 3(89), 2009.

\bibitem{ee:10}
A.~Eldar and M.~B. Elowitz.
\newblock Functional roles for noise in genetic circuits.
\newblock {\em Nature}, 467:167--173, 2010.

\bibitem{EK86}
S.~N. Ethier and T.~G. Kurtz.
\newblock {\em Markov processes}.
\newblock Wiley Series in Probability and Mathematical Statistics: Probability
  and Mathematical Statistics. John Wiley \& Sons, Inc., New York, 1986.
\newblock Characterization and convergence.

\bibitem{fcs:10}
N.~Friedman, L.~Cai, and X.S. Xie.
\newblock Stochasticity in gene expression as observed by single-molecule
  experiments in live cells.
\newblock {\em Israel Journal of Chemistry}, 49:333--342, 2010.

\bibitem{gb:00}
M.A. Gibson and J.~Bruck.
\newblock Efficient exact stochastic simulation of chemical systems with many
  species and many channels.
\newblock {\em The Journal of Physical Chemistry A}, 104(9):1876--1889, 2000.

\bibitem{gill:92}
D.T. Gillespie.
\newblock A rigorous derivation of the chemical master equation.
\newblock {\em Physica A}, 188:404--425, 1992.

\bibitem{hlw:06}
E.~Hairer, C.~Lubich, and G.~Wanner.
\newblock {\em Geometric Numerical Integration --- Structure-Preserving
  Algorithms for Ordinary Differential Equations}.
\newblock Springer Series in Computational Mathematics. Springer-Verlag, New
  York, 2 edition, 2006.

\bibitem{hnw:93}
E.~Hairer, S.P. N{\o}rsett, and G.~Wanner.
\newblock {\em Solving Ordinary Differential Equations I: Nonstiff Problems}.
\newblock Springer Series in Computational Mathematics. Springer-Verlag, Berlin
  Heidelberg, second revised editions edition, 1993.

\bibitem{hw:96}
E.~Hairer and G.~Wanner.
\newblock {\em Solving Ordinary Differential Equations II: Stiff and
  Differential-Algebraic Problems}.
\newblock Springer-Verlag, Berlin Heidelberg, second revised edition edition,
  1996.

\bibitem{hr:02}
E.L. Haseltine and J.B. Rawlings.
\newblock Approximate simulation of coupled fast and slow reactions for
  stochastic chemical kinetics.
\newblock {\em Journal of Chemical Physics}, 117(15):6959--6969, 2002.

\bibitem{hwkt:13}
J.~Hasenauer, V.~Wolf, A.~Kazeroonian, and F.~J. Theis.
\newblock Method of conditional moments {(MCM)} for the chemical master
  equation : A unified framework for the method of moments and hybrid
  stochastic-deterministic models.
\newblock {\em Journal of Mathematical Biology}, 2013.

\bibitem{hl:07}
A.~Hellander and P.~L{\"{o}}tstedt.
\newblock Hybrid method for the chemical master equation.
\newblock {\em Journal of Computational Physics}, 227(1):100--122, 2007.

\bibitem{hmmw:10}
T.A. Henzinger, L.~Mikeev, M.~Mateescu, and V.~Wolf.
\newblock Hybrid numerical solution of the chemical master equation.
\newblock In {\em Proceedings of the 8th International Conference on
  Computational Methods in Systems Biology}, 8th international Conference on
  Computational Methods in Systems Biology, pages 55--65. ACM, New York, 2010.

\bibitem{jah:11}
T.~Jahnke.
\newblock On reduced models for the chemical master equation.
\newblock {\em Multiscale Modeling and Simulation}, 9(4):1646--1676, 2011.

\bibitem{jk:12}
T.~Jahnke and M.~Kreim.
\newblock Error bound for piecewise deterministic process modelling stochastic
  reaction systems.
\newblock {\em Multiscale Modeling and Simulation}, 10(4):1119--1147, 2012.

\bibitem{KK13}
H.W. Kang and T.G. Kurtz.
\newblock Separation of time-scales and model reduction for stochastic reaction
  networks.
\newblock {\em Annals of Applied Probability}, 23(2):529--583, 2013.

\bibitem{KK14}
Hye-Won Kang, Thomas~G. Kurtz, and Lea Popovic.
\newblock Central limit theorems and diffusion approximations for multiscale
  {M}arkov chain models.
\newblock {\em Ann. Appl. Probab.}, 24(2):721--759, 2014.

\bibitem{klwklh:09}
E.~Klipp, W.~Liebermeister, C.~Wierling, A.~Kowald, H.~Lehrach, and R.~Herwig.
\newblock {\em Systems Biology: A Textbook}.
\newblock WILEY - VCH Verlag GmbH \& Co. KGaA, 2009.

\bibitem{kp:92}
P.E. Kloeden and E.~Platen.
\newblock {\em Numerical Solution of Stochastic Differential Equations}.
\newblock Springer, Berlin, 1992.

\bibitem{kol:00}
W.~Kolch.
\newblock Meaningful relationship: the regulation of {Ras/ Raf/ MEK/ ERK}
  pathway by protein interactions.
\newblock {\em Biochemical Journal}, 351:289--305, 2000.

\bibitem{kcg:05}
W.~Kolch, M.~Calder, and D.~Gilbert.
\newblock When kinases meet mathematics: the systems biology of mapk
  signalling.
\newblock {\em FEBS Letters,}, 579:1891--1895, 2005.

\bibitem{mlsh:12}
S.~Menz, J.C. Latorre, C.~Sch{\"u}tte, and W.~Huisinga.
\newblock Hybrid stochastic--deterministic solution of the chemical master
  equation.
\newblock {\em Multiscale Modeling and Simulation}, 10(4):1232--1262, 2012.

\bibitem{osvcgk:05}
R.J. Orton, O.E. Sturm, V.~Vyshemirsky, M.~Calder, D.R. Gilbert, and W.~Kolch.
\newblock Computational modelling of the receptor-tyrosine-kinase-activated
  {MAPK} pathway.
\newblock {\em Biochemical Journal}, 392:249--261, 2005.

\bibitem{rs:05}
H.~Rubinfeld and R.~Seger.
\newblock The {ERK} cascade: a prototype of {MAPK} signaling.
\newblock {\em Molecular Technology}, 31:151--174, 2005.

\bibitem{sk:05}
H.~Salis and Y.~Kaznessis.
\newblock Accurate hybrid stochastic simulation of a system of coupled chemical
  or biochemical reactions.
\newblock {\em The Journal of Chemical Physics}, 122:054103, 2005.

\bibitem{sl:04}
H.M. Sauro and B.~Ingalls.
\newblock Conservation analysis in biochemical networks computational issues
  for software writers.
\newblock {\em Biophysical Chemistry}, 109:1--15, 2004.

\bibitem{spfhxnlkpbgkwlfkn:09}
B.~Schoeberl, E.A. Pace, J.B. Fitzgerald, B.D. Harms, L.~Xu, L.~Nie, A.~Kalra,
  V.~Paragas, R.~Bukhalid, V.~Grantcharova, N.~Kohli, K.A. West,
  M.~Leszczyniecka, M.J. Feldhaus, A.J. Kudla, and U.B. Nielsen.
\newblock Therapeutically targeting {ErbB3}: A key node in ligand-induced
  activation of the {ErbB} receptor{-PI3K} axis.
\newblock {\em Science Signaling}, 2(77), 2009.

\bibitem{ts:02}
T.~Tian and J.~Song.
\newblock Mathematical modelling of the {MAP} kinase pathway using proteomic
  datasets.
\newblock {\em PLOS}, 7(8):e42230, 2002.

\bibitem{yxrlx:06}
J.~Yu, J.~Xiao, X.~Ren, K.~Lao, and X.S. Xie.
\newblock Probing gene expression in live cells, one protein molecule at a
  time.
\newblock {\em Science}, 311(5767):1600--1603, 2006.

\end{thebibliography}

\end{document}